\documentclass[conference,10pt]{IEEEtran}
\usepackage[utf8]{inputenc}
\usepackage{amsmath}
\usepackage{amsfonts}
\usepackage{cite}
\usepackage{graphicx}
\usepackage{xcolor}
\usepackage[ruled,linesnumbered]{algorithm2e}
\usepackage{multirow}
\usepackage{adjustbox}
\usepackage{gensymb}

\IEEEoverridecommandlockouts

\newtheorem{lemma}{\bf Lemma}

\newcommand{\h}{{\bf h}}
\newcommand{\I}{{\bf I}}

\newcommand{\tr}{{\tt tr}}
\renewcommand{\a}{{\bf a}}

\newcommand{\bSigma}{\boldsymbol{\Sigma}}
\newcommand{\bLambda}{\boldsymbol{\Lambda}}
\newcommand{\bphi}{\boldsymbol{\phi}}
\newcommand{\bpsi}{\boldsymbol{\psi}}
\newcommand{\bGamma}{\boldsymbol{\Gamma}}

\title{Double Nonstationarity: Blind Extraction of Independent Nonstationary Vector/Component from Nonstationary Mixtures --- Algorithms}

\author{Zbyn\v{e}k Koldovsk\'{y}\authorrefmark{1}, \authorblockN{V\'aclav Kautsk\'y\authorrefmark{1}
, Petr Tichavsk\'y\authorrefmark{3}
}
\authorblockA{\authorrefmark{1}Acoustic Signal Analysis and Processing Group, Faculty of Mechatronics, Informatics, 
and Interdisciplinary Studies,\\ Technical University of Liberec, Czech Republic.}
\authorblockA{\authorrefmark{2}The Czech Academy of Sciences, Institute of Information Theory and Automation, Czech Republic.}
\thanks{This work was supported by 
		The Czech Science Foundation through 
		Project No.~20-17720S, and by the Department of the Navy, Office of Naval Research Global, through Project No.~N62909-19-1-2105.}
}

\begin{document}

\maketitle

\begin{abstract}
In this article, nonstationary mixing and source models are  combined for developing new fast and accurate algorithms for Independent Component or Vector Extraction (ICE/IVE), one of which stands for a new extension of the well-known FastICA. This model allows for a moving source-of-interest (SOI) whose distribution on short intervals can be \mbox{(non-)circular} \mbox{(non-)Gaussian}. A particular Gaussian source model assuming tridiagonal covariance matrix structures is proposed. It is shown to be beneficial in the frequency-domain speaker extraction problem. The algorithms are verified in simulations. In comparison to the state-of-the-art algorithms, they show superior performance in terms of convergence speed and extraction accuracy.
\end{abstract}

\section{Introduction}
\subsection{Topic}
Blind Source Separation (BSS) aims at recovering unobserved signals, called sources, from their mixture without additional knowledge  \cite{comon2010handbook}. This area has been vital in the signal processing and machine learning communities over the last three decades. It is pertinent to situations where particular information about the sources is missing and only general assumptions can be stated. 
When the sources are statistically independent, BSS can be solved through Independent Component Analysis (ICA) \cite{comon1994}. Blind Source Extraction (BSE) is a related problem in which the goal is to extract a particular source of interest (SOI). BSE is motivated by the fact that targeting the SOI may often be considerably more cost-effective than separating all of the sources. A BSE counterpart to ICA is Independent Component Extraction (ICE).

It is also possible to consider multiple mixtures (data-sets) and separate them jointly. Joint BSS (jBSS) is advantageous over BSS in situations when relations/dependencies exist among the sources from different data-sets. Processing one SOI extraction from each mixture jointly, we speak about joint BSE (jBSE). The extension of ICA to jBSS is known as Independent Vector Analysis (IVA); the jBSE counterpart of IVA is Independent Vector Extraction (IVE). This article builds on and contributes to ICE and IVE.

\subsection{State-of-the-art}
The existing ICA/IVA/ICE/IVE algorithms can be categorized based on the statistical model of the sources, referred to as source model, which is used for their development. In general, their goal is to capture various signal features as much as possible. However, mathematical tractability and computational costs must  also be taken into account.

In this overview, we focus on two major source models because the key idea of this paper builds on their combination; a survey of BSS methods beyond these classes can be found, e.g., in \cite{adali2014}. 1) The non-Gaussian model considers each source as a sequence of independently and identically distributed (i.i.d.)   non-Gaussian random variables. 2) The nonstationary model allows for varying variance and, typically, assumes that sources are sequences of independent Gaussian variables whose variances are changing from interval to interval. The combination of these models occurs when non-Gaussianity is taken into account in the nonstationary model, so sources are assumed to be i.i.d. Gaussian or non-Gaussian within the intervals.

Non-Gaussianity-based ICA methods represent algorithms based on mutual information minimization \cite{comon1994}, maximum likelihood estimation (MLE) \cite{pham1997}, neural network-like approaches \cite{bell1995,ica:nadal:1994}, etc.; they were shown to be more or less related with MLE  \cite{cardoso1997}. More advanced methods adapt the source model by matching parametric \cite{koldovsky2006} or non-parametric \cite{ica:boscolo:2004} non-Gaussian distributions to the separated sources. The methods also differ in the optimization approach. For example, there are gradient methods \cite{amari1996}, auxiliary-function-based methods \cite{ono2010} or fixed-point algorithms \cite{hyvarinen1999}.

Non-Gaussianity-based BSE can be accomplished by minimizing the output signal entropy \cite{hyvarinen1999}. ICE is based on a reduced mixing model parameterization, in which one source is treated as the SOI and the others, which are not  subject to separation, as background sources \cite{koldovsky2019TSP}. ICE has been shown equivalent with the minimum entropy approach when the background model is multivariate Gaussian \cite{koldovsky2021fastdiva} and with ICA when it is multivariate non-Gaussian \cite{koldovsky2018a,kautsky2020CRLB}.

In contrast to the non-Gaussian model, the nonstationarity-based one can be identified using second-order statistics (SOS) only \cite{belouchrani1997,pham2001}. Numerous methods are based on the Joint Approximate Diagonalization (JAD) of sample covariance matrices computed on intervals (blocks) of data \cite{pham2001,yeredor2002,tichavsky2009}. For BSE, JAD can be replaced by Joint Block Diagonalization (JBD) where the SOI is represented by one-dimensional subspace that is separated from a hyperplane representing the background \cite{nion2011,tichavsky2012,lahat2016}. 

Similarly, the source models have been applied in jBSS and jBSE to model vector sources, where a vector source consists of corresponding scalar sources, one source per mixture (data-set). In IVA and IVE, this approach entails using multivariate non-Gaussian model distributions that capture internal dependencies among the scalar sources \cite{kim2007,ono2011stable,anderson2014,koldovsky2019TSP,scheibler2019overiva}. 
The nonstationary model can be effectively used for jBSS when the elements of vector sources are correlated \cite{li2009,weiss2018}; see also \cite{lahat2016}.

The non-Gaussian and nonstationary source models have been successfully combined in ICA \cite{kellermann2006icassp, Koldovsky2009,li2010} as well as in IVA \cite{ono2012apsipa}. The recent extensions of IVA known under the umbrella of Independent Low Rank Matrix Analysis (ILRMA) can  also be considered as extensions of this kind  \cite{kitamura2016determined,kitamura2018,shinichi2020}. 

Another classification of BSS methods is based on the assumed model of source mixing, that is, the mixing model. 
In the vast majority of BSS literature, the instantaneous linear mixing model is assumed \cite{comon1994,cardoso1998,hyvarinen2001,comon2010handbook}. The other intensively studied convolutive model is also linear. It is often considered in the Fourier transform domain where it is translated to a set of instantaneous mixtures, which can be treated as the jBSS problem \cite{smaragdis1998,makino2007,ASSSEbook2018}. Some specific nonlinear mixing models have been studied, e.g., in \cite{ehsandoust2017,deville2021nonlinear}.

Similarly to source models, nonstationarity can be brought into the mixing models\footnote{Nonstationary mixing models are sometimes termed as ``dynamic models''  \cite{koldovsky2021fastdiva}.}. The goal is to capture the time-variant mixing conditions caused, e.g., by source movements or similar changes. Typically, estimation methods for the static linear mixing are turned into adaptive algorithms 
\cite{taniguchi2014,hsu2016}. The nonstationary mixing process is less frequently described by a more specific parameterization such as that used in \cite{yeredor2003,weisman2006}. Recently, semi-time-variant models denoted as CMV and CSV (Constant Mixing/Separating Vector) have been considered in \cite{kautsky2020CRLB,koldovsky2021fastdiva,jansky2022}. CMV and CSV are designed for BSE/jBSE in which the SOI is static or moving, respectively, on a dynamic background. The nonstationarity is arranged through allowing specific parameters to be changing from interval to interval. In \cite{koldovsky2021fastdiva}, the well-known FastICA algorithm \cite{hyvarinen1999} has been extended for CSV and named as FastDIVA (Fast Dynamic IVA).

\subsection{Contribution}
Although FastDIVA builds on the non-Gaussian model, it partly allows for source nonstationarity. This is because the variance of the SOI is allowed to change over the intervals of the CSV mixing model. However, the performance analysis of FastDIVA in \cite{koldovsky2021fastdiva} as well as the Cram\'er-Rao bound in \cite{kautsky2020CRLB} has shown that the SOI is not identifiable when its distribution is Gaussian. Dividing data into more (shorter) intervals does not seem very effective since the number of mixing parameters proportionally grows, and, moreover, the Gaussian SOI remains unidentifiable. 

This is the main motivation behind the novel extension  provided in this article: We propose to combine the nonstationary CSV mixing model with the nonstationary source model, enabling Gaussian and/or non-Gaussian moving SOI. Two second-order algorithms are derived, one of which stands for a new extension of FastDIVA (resp. FastICA and FastIVA). Special attention is given to the Gaussian source model, for which distinguished SOS-based variants of the algorithms are derived. The latter can also efficiently benefit  from the SOI non-circularity. Moreover, a particular Gaussian source model for the SOI, which assumes tridiagonal covariance matrix structures, is proposed. This model is efficiently implemented within the proposed algorithms, and it shows promising results in the frequency-domain speaker extraction problem where the $K$-dimension (the number of frequencies/mixtures/data-sets) can take value in the order of hundreds. These methods are verified by extensive numerical studies. In comparison to state-of-the-art algorithms, they show superior performance in terms of convergence speed and extraction accuracy.

The paper is organized as follows. Technical description of the problem is formulated in Section~II. In Section~III, the second-order algorithms are derived. Section~IV is devoted to special variants of the algorithms based on the Gaussian source model. Experimental validation is provided in Section~V; and Section~VI concludes the paper.

\section{Problem Formulation}
We consider measurements of length $N$ in $K$ data-sets, each one obtained by $d$ sensors, i.e., of dimension $d$. Each measurement is divided into $T$ non-overlapping intervals of length $N_b$, hereafter called blocks, and each block is divided into $L$ sub-blocks of length $N_s$. For simplicity, the blocks (and sub-blocks) have the same length, although different lengths could be considered as well. Hence, $N=T\cdot N_b$, $N_b=L\cdot N_s$, and $N=T\cdot L \cdot N_s$. Throughout this article, the index of a data-set, block, and sub-block, will always be denoted, respectively, by $k=1,\dots,K$, $t=1,\dots,T$, and $\ell=1,\dots,L$. 

\subsection{Semi-time-variant mixing model}
The $n$th sample of the measured data, $n=1,\dots,N_s$, within the $\ell$th sub-block of the $t$th block and in the $k$th data-set is modeled as linear instantaneous mixture
\begin{equation}\label{eq:mixingmodel}
    {\bf x}_{k,t,\ell}(n)={\bf A}_{k,t}{\bf u}_{k,t,\ell}(n),
\end{equation}
where the mixing matrix ${\bf A}_{k,t}$ is parameterized as
\begin{equation}\label{eq:mixingmatrix}
    {\bf A}_{k,t}= 
    \begin{pmatrix}
 {\bf a}_{k,t} & {\bf Q}_{k,t}
\end{pmatrix}  =
 \begin{pmatrix}
  \gamma_{k,t} & {\bf h}_{k}^H\\
   {\bf g}_{k,t} &  \frac{1}{\gamma_{k,t}}({\bf g}_{k,t}{\bf h}_{k}^H-\I_{d-1})
    \end{pmatrix}.
\end{equation}
The signal samples will be assumed i.i.d. within the sub-blocks; therefore, the argument $n$ can be omitted;
${\bf u}_{k,t,\ell}=[s_{k,t,\ell};\, {\bf z}_{k,t,\ell}]$ is the vector of the source components $s_{k,t,\ell}$ and ${\bf z}_{k,t,\ell}$ representing, respectively, the SOI and background; let their mean value be zero; ${\bf a}_{k,t}$ is the mixing vector (the first column of ${\bf A}_{k,t}$) corresponding to the SOI; ${\bf I}_d$ denotes the $d\times d$ identity matrix. Note that \eqref{eq:mixingmodel} can also be written in the form
\begin{equation}\label{eq:mixingmodel_alternative}
   {\bf x}_{k,t,\ell}= {\bf a}_{k,t}s_{k,t,\ell} + {\bf y}_{k,t,\ell},
\end{equation}
where ${\bf y}_{k,t,\ell}={\bf Q}_{k,t}{\bf z}_{k,t,\ell}$ play the role of the background signals as they are observed by the sensors.

In the sequel, ${\bf s}_{t,\ell}=(s_{1,t,\ell},\dots,s_{K,t,\ell})^T$ will refer to the vector component of the SOI; $s_{k,t,\ell}$ and ${\bf z}_{k,t,\ell}$ will be, respectively, called the $k$th component of the SOI and of the background.

Equivalently to \eqref{eq:mixingmodel}, the de-mixing model reads
\begin{equation}
    {\bf u}_{k,t,\ell}={\bf W}_{k,t}{\bf x}_{k,t,\ell},
\end{equation}
where
\begin{equation}\label{eq:demixingmodel}
    {\bf W}_{k,t} =
     \begin{pmatrix}
     {\bf w}_{k}^H\\
     {\bf B}_{k,t}
     \end{pmatrix}  =
     \begin{pmatrix}
     \beta_{k}^* & {\bf h}_{k}^H\\
     {\bf g}_{k,t} & -\gamma_{k,t} \I_{d-1}
     \end{pmatrix},
\end{equation}
where ${\bf w}_{k}=[\beta_{k};{\bf h}_{k}]$ is the separating vector, and ${\bf a}_{k,t}$ and ${\bf w}_{k}$ are assumed to satisfy the distortionless constraint ${\bf w}_{k}^H{\bf a}_{k,t}=1$. Under this constraint, the reader can easily verify that ${\bf W}_{k,t}$ in (\ref{eq:demixingmodel}) is the inverse matrix of ${\bf A}_{k,t}$ in (\ref{eq:mixingmodel}). It holds that $\det{\bf W}_{k,t}=(-1)^{d-1}\gamma_{k,t}^{d-2}$; see Eq. (15) in \cite{koldovsky2019TSP}. Note that the separating vectors ${\bf w}_{k}$ are independent of $t$ while the mixing vectors ${\bf a}_{k,t}$ depend on it. This parameterization corresponds to the semi-time-variant CSV mixing model advocated in  \cite{kautsky2020CRLB,koldovsky2021fastdiva,jansky2022}.

\subsection{Source model}\label{sec:sourcemodel}
It is worth pointing out that the mixing parameters remain constant within the blocks while the source model is i.i.d. (stationary) only within the sub-blocks. This means that signals are allowed to be more dynamic than  the changes in the mixing process; this approach is more suitable for real situations.

As for the SOI, $s_{1,t,\ell},\dots,s_{K,t,\ell}$ are modeled jointly, and their the joint probability density function (pdf) is denoted by $p_{t,\ell}({\bf s}_{t,\ell})$. The idea of joint statistical modeling is adopted from IVA. It allows for mutual dependencies among the components of the SOI, which helps in solving the permutation problem \cite{kim2006}. Our extension here is that the pdf is allowed to vary across blocks and sub-blocks (dependent on $t$ and $\ell$). Since $p_{t,\ell}({\bf s}_{t,\ell})$ is not known, it was proposed in \cite{koldovsky2021fastdiva} that an appropriate surrogate is
\begin{equation}\label{eq:modeldensity}
    p_{t,\ell}({\bf s}_{t,\ell}) \approx f\left(\left\{\frac{s_{k,t,\ell}}{\hat\sigma_{k,t,\ell}}\right\}_k\right)\left(\prod_{k=1}^K\hat\sigma_{k,t,\ell}\right)^{-2},
\end{equation}
where $f(\cdot)$ is a suitable normalized pdf\footnote{The model density $f(\cdot)$ could have been considered as dependent on $t$ and $\ell$. However, since there is typically lack of information about the pdf of the SOI, we find it more practical when $f(\cdot)$ is constant and the variability of the pdf is captured only by the time-varying variance $\hat\sigma_{k,t,\ell}^2$.}, and $\hat\sigma_{k,t,\ell}^2$ is the sample-based variance of the estimate of $s^{k,t,\ell}$.

The background probabilistic model is assumed circular Gaussian, namely, ${\bf z}^{k,t,\ell}\sim\mathcal{CN}({\bf 0},{\bf C}_{\bf z}^{k,t,\ell})$ where ${\bf C}_{\bf z}^{k,t,\ell}$ is an unknown covariance matrix. Note that this involves the assumption that the background signals from different data-sets are uncorrelated (hence, owing to the Gaussianity, also independent). 

The fact that the non-Gaussianity, non-circularity, and the dependencies among the components of the background signals are not assumed, brings about important simplifications into algorithms and bounds derivations. As it has been observed with similar problems, the probability model mismatch  does not usually cause algorithm malfunction. Typically, the price for the simplification is a suboptimality in terms of the theoretical achievable extraction accuracy \cite{koldovsky2018a,koldovsky2019TSP}.

\subsection{Contrast function}
The contrast function is a function of the mixing parameters whose optimum points provide their consistent estimates. The function is derived from the likelihood function by replacing unknown pdfs and nuisance parameters; it is sometimes referred to as the quasi-likelihood function \cite{pham1997,weiss2018}.

By comparing the mixing and source models with the one in \cite{koldovsky2021fastdiva}, the model presented here differs only in that $\sigma_{k,t,\ell}^2$ and ${\bf C}_{\bf z}^{k,t,\ell}$ are allowed to be changing over the sub-blocks. This allows us to obtain the contrast function by straightforward modifications of Eq. (12) in \cite{koldovsky2021fastdiva}.

Namely, the signals and their parameters become dependent on the sub-block index $\ell$. Therefore, the sample-average operator denoted by $\hat{\rm E}[\cdot]$ computes the average only over the samples in the sub-block (not over the block as in \cite{koldovsky2021fastdiva}). The mixing parameters remain the same, i.e., independent of $\ell$. Finally, the whole formula must be averaged over the sub-blocks. Therefore, the contrast function takes on the form
\begin{multline}\label{eq:contastfull}
    \mathcal{C}\left(\{{\bf w}_k,{\bf a}_{k,t}\}_{k,t}\right) =\Bigg<\Bigg<\hat{\rm E}\left[\log f\left(\left\{\frac{\hat{s}_{k,t,\ell}}{\hat\sigma_{k,t,\ell}}\right\}_k\right)\right]  \\ -\sum_{k=1}^{K}\log\hat\sigma_{k,t,\ell}^2 -\sum_{k=1}^{K} \hat{\rm E}\left[\hat{\bf z}_{k,t,\ell}^H({\bf C}_{\bf z}^{k,t,\ell})^{-1}\hat{\bf z}_{k,t,\ell}\right]\Bigg>_\ell  \\ 
    + (d-2)\sum_{k=1}^{K}\log |\gamma_{k,t}|^2\Bigg>_t,
\end{multline}
where $\hat{s}_{k,t,\ell}={\bf w}_k^H{\bf x}_{k,t,\ell}$ is the estimate of the SOI, $\hat{\bf z}_{k,t,\ell}={\bf B}_k{\bf x}_{k,t,\ell}$ is the estimated background, and $\hat\sigma_{k,t,\ell}^2$ denotes the sample-based variance of $\hat{s}_{k,t,\ell}$. The operators $\left<\cdot\right>_t$ and $\left<\cdot\right>_\ell$ denote averaging over the index $t$ and $\ell$, respectively.

\section{Proposed Algorithms}
In this section, we derive second-order derivative-based algorithms seeking for the desired optimum point of the contrast function \eqref{eq:contastfull}. For their development, we make use of the fact that the terms in \eqref{eq:contastfull} are mostly separated. Therefore, we simplify the exposition as if $K=1$ and $T=1$; so the indices $k$ and $t$ can be dropped. The extension to $K>1$ and $T>1$ will be discussed later. 

For $K=1$ and $T=1$, we have the time-invariant instantaneous one-mixture problem studied under the umbrella of ICE \cite{koldovsky2019TSP}. The contrast function is simplified to
\begin{multline}\label{eq:contastICE}
    \mathcal{C}_1\left({\bf w},{\bf a}\right) =\Bigg<\hat{\rm E}\left[\log f\left(\frac{\hat{s}^\ell}{\hat\sigma_\ell}\right)\right]  -\log\hat\sigma_\ell^2 -\hat{\rm E}\left[\hat{\bf z}_\ell^H{\bf R}_\ell\hat{\bf z}_\ell\right]\Bigg>_\ell \\ 
    + (d-2)\log |\gamma|^2, 
\end{multline}
where we have introduced auxiliary matrices ${\bf R}_\ell$, whose ideal value is ${\bf R}_\ell=({\bf C}_{\bf z}^{\ell})^{-1}$. Since ${\bf C}_{\bf z}^{\ell}$ are not known, we select the value of ${\bf R}_\ell$ later in Lemma~1.

\subsection{Orthogonal constraint}
The parameter vectors ${\bf w}$ and ${\bf a}$ are almost free, linked only through the distortionless constraint ${\bf w}^H{\bf a}=1$. Since the contrast function has many spurious extremes where ${\bf w}$ and ${\bf a}$ do not correspond to the same source, it is helpful to link ${\bf w}$ and ${\bf a}$ more tightly using the orthogonal constraint (OGC).

By definition, the OGC requires that the sample correlations of the estimated SOI and background be zero. When imposing the OGC for each sub-block, that is,
\begin{equation}\label{eq:OGC_on_subblocks}
    \hat{\rm E}[\hat{s}_\ell^*\hat{\bf z}_\ell]=0,\qquad \ell=1,\dots,L,
\end{equation} 
we have $L(d-1)$ conditions, which, together with the distortionless constraint ${\bf w}^H{\bf a}=1$, provides $L(d-1)+1$ linear conditions on ${\bf a}$ (or ${\bf w}$) in total. They cannot in general be satisfied all simultaneously  unless $L=1$. 

Since the case $L>1$ is of particular interest in this work, we propose to replace \eqref{eq:OGC_on_subblocks} by a weaker condition
\begin{equation}\label{eq:OGC_over_subblocks}
    \left<\hat{\rm E}[\hat{s}_\ell^*\hat{\bf z}_\ell]\right>_\ell=0,
\end{equation} 
which imposes orthogonality over the whole block of signals (ignoring sub-blocks). When ${\bf a}$ is treated as the dependent variable, we can apply the formula derived in Appendix A in \cite{koldovsky2019TSP}, and the solution of \eqref{eq:OGC_over_subblocks} satisfying ${\bf w}^H{\bf a}=1$ is 
\begin{equation}\label{eq:OGC}
    {\bf a}=\frac{\widehat{\bf C}{\bf w}}{{\bf w}^H\widehat{\bf C}{\bf w}},
\end{equation}
where $\widehat{\bf C}=\left<\widehat{\bf C}_\ell\right>_\ell$ and $\widehat{\bf C}_\ell=\hat{\rm E}[{\bf x}_\ell{\bf x}_\ell^H]$ is the sample covariance matrix of ${\bf x}_\ell$.

\subsection{Gradient}\label{sec:gradient}
The first step for deriving the algorithms is to compute the gradient of \eqref{eq:contastICE} with respect to ${\bf w}^H$ when ${\bf a}$ is dependent through \eqref{eq:OGC}.
We summarize the result in the following Lemma.

\begin{lemma}
Let, after computing the derivatives, the matrices ${\bf R}_\ell$ be put equal to ${\bf R}_\ell=\left< \widehat{\bf C}_{\bf z}^{\ell} \right>_\ell^{-1}$ where $\widehat{\bf C}_{\bf z}^{\ell}=\hat{\rm E}[\hat{\bf z}_\ell\hat{\bf z}_\ell^H]$ is the sample-based covariance matrix of ${\bf z}_\ell$. It then holds that
\begin{multline}\label{eq:gradient_complete}
    \frac{\partial}{\partial {\bf w}^H}\mathcal{C}_1\left({\bf w},\frac{\widehat{\bf C}{\bf w}}{{\bf w}^H\widehat{\bf C}{\bf w}}\right)=\\{\bf a} -\Bigg<\hat{\rm E}\left[\phi\left(\frac{\hat{s}_\ell}{\hat\sigma_\ell}\right)\frac{{\bf x}_\ell}{\hat\sigma_\ell}\right]+\Re(\hat\nu_\ell){\bf a}_\ell - {\bf a}_\ell\Bigg>_\ell,
\end{multline}
where $\Re(\cdot)$ denotes the real part of the argument, 
\begin{equation}\label{eq:score}
\phi(s) = -\frac{\partial}{\partial {s}} \log f(s)
\end{equation}
is the score function of the model density $f(s)$,
$\hat\nu_\ell$ is the sample-based estimate of
\begin{equation}\label{eq:nu}
\nu_\ell = {\rm E}\left[\phi\left(\frac{{s}_\ell}{\sigma_\ell}\right)\frac{s_\ell}{\sigma_\ell}\right],
\end{equation}
and
\begin{equation}\label{eq:OGCsubblock}
    {\bf a}_\ell =\frac{\widehat{\bf C}_\ell{\bf w}}{{\bf w}^H\widehat{\bf C}_\ell{\bf w}}.
\end{equation}
\end{lemma}
\begin{proof}
See Appendix~A.
\end{proof}

We now need to make an adjustment of the model density $f(\cdot)$ for the sake of consistency. The problem is revealed by the following Lemma. Hereafter,  ${\bf w}^\star$ will denote the true separating vector such that $({\bf w}^\star)^H{\bf x}_\ell=s_\ell$, $\ell=1,\dots,L$.

\begin{lemma}
Let ${\bf w}={\bf w}^\star$ and $N\rightarrow +\infty$. Then, the right-hand side of \eqref{eq:gradient_complete} converges to
\begin{equation}\label{eq:consistency_condition}
    \bigl(2-\left<\Re(\nu_\ell)+\nu_\ell\right>_\ell\bigr)\, {\bf a}.
\end{equation}
\end{lemma}
\begin{proof}
For $N\rightarrow +\infty$, it holds that ${\bf a}_\ell\rightarrow {\bf a}$, $\hat\nu_\ell\rightarrow \nu_\ell$, and by \eqref{eq:mixingmodel_alternative} and using the fact that $s_\ell$ and ${\bf y}_\ell$ are independent,
\begin{equation}
    \hat{\rm E}\left[\phi\left(\frac{{s}_\ell}{\hat\sigma_\ell}\right)\frac{{\bf x}_\ell}{\hat\sigma_\ell}\right]\stackrel{N\rightarrow+\infty}{\longrightarrow} {\rm E}\left[\phi\left(\frac{{s}_\ell}{\sigma_\ell}\right)\frac{{\bf a}s_\ell+{\bf y}_\ell}{\sigma_\ell}\right]=\nu_\ell{\bf a}.
\end{equation}
The assertion of the lemma follows.
\end{proof}

To make the stationary point of the contrast function a consistent estimate of ${\bf w}$, \eqref{eq:consistency_condition} must be equal to zero when ${\bf w}={\bf w}^\star$ and $N\rightarrow +\infty$. As observed in previous works \cite{koldovsky2019TSP,koldovsky2021fastdiva}, this problem appears due to the arbitrarily chosen model density $f(\cdot)$. It can easily be  solved by considering a suitable sub-block-dependent modification of $f(\cdot)$. This step is performed through the substitution $\phi(\cdot)\rightarrow\hat\nu_\ell^{-1}\phi(\cdot)$. Consequently, the gradient \eqref{eq:gradient_complete} turns to
\begin{equation}\label{eq:normalizedgrad}
    \nabla={\bf a}-\Bigg<\hat\nu_\ell^{-1}\hat{\rm E}\left[\phi\left(\frac{\hat{s}_\ell}{\hat\sigma_\ell}\right)\frac{{\bf x}_\ell}{\hat\sigma_\ell}\right]\Bigg>_\ell.
\end{equation}
The reader can verify that, for ${\bf w}={\bf w}^\star$ and $N\rightarrow +\infty$, $\nabla={\bf 0}$ holds, which ensures consistency.

\subsection{Hessian}\label{sec:hessian}
The Hessian matrices of the real-valued contrast function defined using the Wirtinger calculus are given by \cite{li2008}
\begin{align}
{\bf H}_1^{\rm expl}&=\frac{\partial^2\mathcal{C}_1}{\partial{\bf w}^T\partial{\bf w}}=\frac{\partial\nabla^H}{\partial{\bf w}},\\
{\bf H}_2^{\rm expl}&=\frac{\partial^2\mathcal{C}_1}{\partial{\bf w}^H\partial{\bf w}}=\frac{\partial\nabla^T}{\partial{\bf w}}.
\end{align}
The superscript $^{\rm expl}$ is used to distinguish the explicit  Hessian matrices from their counterparts that are finally used in the algorithms. These are obtained by considering the analytical shapes of ${\bf H}_1^{\rm expl}$ and ${\bf H}_2^{\rm expl}$ when ${\bf w}={\bf w}^\star$ and $N\rightarrow +\infty$. 

We consider two approaches: the derivatives of \eqref{eq:normalizedgrad} are considered with and without the imposed OGC on ${\bf a}$, respectively. In both computations, the $\hat\nu_\ell$s variables are treated as constants. The results are summarized in the following two Lemmas.

\begin{lemma}\label{lemma:hessianfastdiva}
Let ${\bf w}={\bf w}^\star$, $N\rightarrow +\infty$, $\hat\nu_\ell$ be constants, and ${\bf a}$ depend on ${\bf w}$ through \eqref{eq:OGC}. It then holds that ${\bf H}_1^{\rm expl}\rightarrow {\bf H}_1$ and ${\bf H}_2^{\rm expl}\rightarrow {\bf H}_2$, where
\begin{align}
    {\bf H}_1^*&=\left<\nu_\ell^{-1}\Bigl(\frac{\tau_\ell}{2}{\bf a}_\ell-(\nu_\ell+\eta_\ell){\bf a}\Bigr)  \right>_\ell{\bf a}^T,\label{eq:H1fast}\\
    {\bf H}_2&=\frac{\bigl<{\bf C}_\ell^*\bigr>_\ell}{\bigl<\sigma_\ell^2\bigr>_\ell}-\left<\frac{\rho_\ell{\bf C}_\ell^*}{\nu_\ell\sigma_\ell^2}\right>_\ell -\left<\nu_\ell^{-1}\left(\omega_\ell{\bf a}^*-\frac{\tau_\ell}{2}{\bf a}_\ell^*\right)\right>_\ell{\bf a}^T,\label{eq:H2fast}
\end{align}
where $\tau_\ell=\eta_\ell+\xi_\ell+\nu_\ell$ and $\omega_\ell=\xi_\ell+\nu_\ell-\rho_\ell$, and
\begin{align}
    \rho_\ell &= {\rm E}\left[\frac{\partial\phi(\frac{s_\ell}{\sigma_\ell})}{\partial s^*} \right],\label{eq:rho}\\
    \xi_\ell &= {\rm E}\left[\frac{\partial\phi(\frac{s_\ell}{\sigma_\ell})}{\partial s^*}\frac{|s_\ell|^2}{\sigma_\ell^2} \right],\label{eq:xi}\\
    \eta_\ell &= {\rm E}\left[\frac{\partial\phi(\frac{s_\ell}{\sigma_\ell})}{\partial s}\frac{s_\ell^2}{\sigma_\ell^2} \right]\label{eq:eta}.
\end{align}
\end{lemma}
\begin{proof}
See Appendix~B.
\end{proof}

\begin{lemma}\label{lemma:hessianquickive}
Let ${\bf w}={\bf w}^\star$, $N\rightarrow +\infty$, and $\hat\nu_\ell$ and ${\bf a}$ be constants. It then holds that  ${\bf H}_1^{\rm expl}\rightarrow {\bf H}_1$ and ${\bf H}_2^{\rm expl}\rightarrow {\bf H}_2$, where
\begin{align}
    {\bf H}_1^*&=\left<\nu_\ell^{-1}\Bigl(\frac{\tau_\ell}{2}{\bf a}_\ell-\eta_\ell{\bf a}\Bigr)  \right>_\ell{\bf a}^T,\label{eq:H1quick}\\
    {\bf H}_2&=-\left<\frac{\rho_\ell{\bf C}_\ell^*}{\nu_\ell\sigma_\ell^2}\right>_\ell -\left<\nu_\ell^{-1}\left((\xi_\ell-\rho_\ell){\bf a}^*-\frac{\tau_\ell}{2}{\bf a}_\ell^*\right)\right>_\ell{\bf a}^T.\label{eq:H2quick}
\end{align}
\end{lemma}
\begin{proof}
See Appendix~B.
\end{proof}

\subsection{Learning rule}
The learning rule in the proposed algorithms is inspired by the exact Newton-Raphson (NR) update derived in \cite{li2008}. An iteration of the exact NR algorithm is given by
\begin{equation}\label{eq:newtonraphsonupdate}
    {\bf w}^{\rm new} = {\bf w} - {\bf H}_3^{-1}(\nabla-({\bf H}^{\rm expl}_1)^*({\bf H}^{\rm expl}_2)^{-1}{\nabla}^*),
\end{equation}
where ${\bf H}_3=({\bf H}_2^{\rm expl})^*-({\bf H}^{\rm expl}_1)^*({\bf H}^{\rm expl}_2)^{-1}{\bf H}^{\rm expl}_1$. We employ this update with the following two modifications:
\begin{enumerate}
    \item ${\bf H}^{\rm expl}_1$ and ${\bf H}^{\rm expl}_2$ are replaced, respectively, by ${\bf H}_1$ and ${\bf H}_2$, in which the unknown signal statistics $\nu_\ell$, $\rho_\ell$, \dots are replaced by their sample-based estimates using samples of the current estimate of the SOI, and
    \item the rank-one terms in ${\bf H}_1$ and ${\bf H}_2$ are neglected, hence, the entire ${\bf H}_1$ is put equal to zero.
\end{enumerate}
After these modifications, the update rule \eqref{eq:newtonraphsonupdate} is simplified to
\begin{equation}\label{eq:simpleupdate}
    {\bf w}^{\rm new} = {\bf w} - {\bf H}^{-1}\nabla,
\end{equation}
where $\nabla$ is computed the same as in \eqref{eq:normalizedgrad} and
\begin{equation}
   {\bf H} =  \frac{\bigl<\widehat{\bf C}_\ell\bigr>_\ell}{\bigl<\hat\sigma_\ell^2\bigr>_\ell}-\left<\frac{\hat\rho_\ell\widehat{\bf C}_\ell}{\hat\nu_\ell^*\hat\sigma_\ell^2}\right>_\ell \label{eq:Hfastdiva}
\end{equation}
for the first proposed algorithm based on Lemma~\ref{lemma:hessianfastdiva}, and 
\begin{equation}
   {\bf H} =  -\left<\frac{\hat\rho_\ell\widehat{\bf C}_\ell}{\hat\nu_\ell^*\hat\sigma_\ell^2}\right>_\ell \label{eq:Hquickive}
\end{equation}
for the second proposed algorithm based on Lemma~\ref{lemma:hessianquickive}. Both algorithms are iterated according to the scheme given by Algorithm~\ref{algorithm:general} until convergence prevails, where the same stopping rule (line~6 in Algorithm~\ref{algorithm:general}) is used as that in \cite{hyvarinen1999}.
For the sake of consistency with our previous works \cite{koldovsky2021fastdiva,koldovsky2021quickive}, the algorithms will be referred to as FastDIVA and QuickIVE, respectively.

Our neglecting the rank-1 terms in ${\bf H}_1$ and ${\bf H}_2$ is justified by the fact that we do not observe any practical improvement when these terms are kept. Similar simplification has been used in \cite{koldovsky2021fastdiva} when $L=T=1$, where it is justified by Proposition~2 in a mathematically rigorous way.

\begin{algorithm}[t]
	\label{algorithm:general}
	\caption{General scheme of the proposed algorithms for blind source extraction}
	\SetAlgoLined
	\KwIn{${\bf x}$, ${\bf w}_{\rm ini}$, ${\tt tol}$}
	\KwOut{${\bf a},{\bf w}$}
	${\bf w}={\bf w}_{\rm ini}$\\
	\Repeat{${\rm crit}<{\tt tol}$}{
	${\bf w}_{\rm old}={\bf w}$;\\
	Update ${\bf a}$ according to \eqref{eq:OGC};\\
	Update ${\bf w}$ according to \eqref{eq:simpleupdate};\\
	${\rm crit}=1-\frac{|{\bf w}^H{\bf w}_{\rm old}|}{\|{\bf w}\|\|{\bf w}_{\rm old}\|}$;
	}
\end{algorithm}

\subsection{Extension to $T>1$ and $K>1$}
We now can get back to the original notation with all three indices $k$, $t$, and $\ell$ (data-set, block, and sub-block) and admit that $T>1$ and $K>1$. The extension of the results derived in Sections~\ref{sec:gradient} and \ref{sec:hessian} to $T>1$ is straightforward because signals' samples are assumed to be independently distributed across the blocks. For the extension to $K>1$, the only difference is that $f(\cdot)$ in \eqref{eq:contastfull} is a function of the entire vector component of the SOI; this fact must be reflected when extending the definitions \eqref{eq:nu}, \eqref{eq:rho}-\eqref{eq:eta}. Let the $k$th score function related to $f(\cdot)$ be defined as
\begin{equation}
    \label{eq:scorek}
\phi_k({\bf s}_{t,\ell}) = -\frac{\partial}{\partial s_k} \log f({\bf s}_{t,\ell}),
\end{equation}
where the partial derivative is taken over the $k$th argument denoted by $s_k$. We can now continue with the new definitions
\begin{align}
    \nu_{k,t,\ell} &= {\rm E}\left[\phi_k\left(\left\{\frac{{s}_{k,t,\ell}}{\sigma_{k,t,\ell}}\right\}_k\right)\frac{s_{k,t,\ell}}{\sigma_{k,t,\ell}}\right]\label{eq:nuktl},\\
    \rho_{k,t,\ell} &= {\rm E}\left[\frac{\partial\phi_k\left(\left\{\frac{{s}_{k,t,\ell}}{\sigma_{k,t,\ell}}\right\}_k\right)}{\partial s_k^*} \right],\label{eq:rhoktl}\\
    \xi_{k,t,\ell} &= {\rm E}\left[\frac{\partial\phi_k\left(\left\{\frac{{s}_{k,t,\ell}}{\sigma_{k,t,\ell}}\right\}_k\right)}{\partial s_k^*}\frac{|s_{k,t,\ell}|^2}{\sigma_{k,t,\ell}^2} \right],\label{eq:xiktl}\\
    \eta_{k,t,\ell} &= {\rm E}\left[\frac{\partial\phi_k\left(\left\{\frac{{s}_{k,t,\ell}}{\sigma_{k,t,\ell}}\right\}_k\right)}{\partial s_k}\frac{s_{k,t,\ell}^2}{\sigma_{k,t,\ell}^2} \right]\label{eq:etaktl}.
\end{align}
The averaging operator $\left<\cdot\right>_t$ in  \eqref{eq:contastfull} causes the gradient and the Hessian matrices to be equal to the averages of their counterparts evaluated on blocks; they obviously depend on $k$. 

Finally, the update rules for $K>1$ and $T>1$ are given by
\begin{equation}\label{eq:simpleupdatektl}
    {\bf w}^{\rm new}_k = {\bf w}_k - {\bf H}_k^{-1}\nabla_k,\qquad k=1,\dots,K,
\end{equation}
where
\begin{equation}\label{eq:fullgradient}
    \nabla_k=\left<{\bf a}_{k,t}-\Bigg<\hat\nu_{k,t,\ell}^{-1}\hat{\rm E}\left[\phi_k\left(\left\{\frac{\hat{s}_{k,t,\ell}}{\hat\sigma_{k,t,\ell}}\right\}_{k}\right)\frac{{\bf x}_{k,t,\ell}}{\hat\sigma_{k,t,\ell}}\right]\Bigg>_\ell\right>_t,
\end{equation}
\begin{equation}
       {\bf a}_{k,t}=\frac{\bigl<\widehat{\bf C}_{k,t,\ell}\bigr>_\ell{\bf w}_k}{{\bf w}_k^H\bigl<\widehat{\bf C}_{k,t,\ell}\bigr>_\ell{\bf w}_k}, 
\end{equation}
and
\begin{align}
   {\bf H}_k &=  \left<\frac{\bigl<\widehat{\bf C}_{k,t,\ell}\bigr>_\ell}{\bigl<\hat\sigma_{k,t,\ell}^2\bigr>_\ell}-\left<\frac{\hat\rho_{k,t,\ell}\widehat{\bf C}_{k,t,\ell}}{\hat\nu_{k,t,\ell}^*\hat\sigma_{k,t,\ell}^2}\right>_\ell\right>_t, \label{eq:Hfastdivaktl}\\
   {\bf H}_k &= -\left<\left<\frac{\hat\rho_{k,t,\ell}\widehat{\bf C}_{k,t,\ell}}{\hat\nu_{k,t,\ell}^*\hat\sigma_{k,t,\ell}^2}\right>_\ell \right>_t \label{eq:Hquickivektl},
\end{align}
for FastDIVA and QuickIVE, respectively. 

\subsection{Relation to previous methods}
For $L=1$, the update rule \eqref{eq:simpleupdatektl} with \eqref{eq:fullgradient} and \eqref{eq:Hfastdivaktl} is readily simplified to Eq. 45 in \cite{koldovsky2021fastdiva}. It means that FastDIVA proposed in this paper is an extension of the previous method for $L>1$. It also follows that the proposed algorithm is the successor of FastICA from   \cite{hyvarinen1999} (only $L=T=K=1$) and of FastIVA from \cite{lee2007fast} (only $L=T=1$). 

Similarly, the proposed QuickIVE is the extension of the method from \cite{koldovsky2021quickive} for $L>1$. QuickIVE provides an alternative to FastDIVA. It is an algorithm whose convergence is slightly slower than that of FastDIVA; nevertheless, this algorithm sometimes appear to be more stable, as will be shown in Section~V; see also \cite{koldovsky2021quickive} where QuickIVE is shown to take an advantage over FastDIVA in continuous on-line source extraction.

\section{Extensions}
The current section is devoted to the Gaussian SOI source model, which comes into play when $L>1$ (unlike for $L=1$, the Gaussian SOI can be identified when $L>1$). The Gaussian source model brings two important advantages. First, its analytic form leads to simplifications of both mathematical expressions and  algorithms, which then operate purely with the second-order statistics: covariance (and pseudo-covariance) matrices. This is useful for capturing non-circularity and dependencies among the components of SOI when $K>1$. Second, dealing with estimated covariance matrices opens up new possibilities for solving difficult situations with a critical lack of data due to very short sub-blocks, i.e., $K\gg N_s$.

\subsection{Gaussian score function}
Let the distribution of ${\bf s}_{t,\ell}$ be Gaussian with covariance matrix $\bSigma_{t,\ell}={\rm E}[{\bf s}_{t,\ell}{\bf s}_{t,\ell}^H]$ and pseudo-covariance matrix $\bGamma_{t,\ell}={\rm E}[{\bf s}_{t,\ell}{\bf s}_{t,\ell}^T]$. From now on, we omit the indices $k$ and $t$ to simplify our notation, keeping in mind that the signals and their parameters are always block- and sub-block-dependent. 

The log-density of the Gaussian SOI, represented by the vector ${\bf s}$, can be written in the form
\begin{equation}
    \log f({\bf s}|\bSigma,\bGamma)=-{\bf s}^H{\bf P}^{-*}{\bf s} + \Re\left\{{\bf s}^T{\bf M}^T{\bf P}^{-*}{\bf s}\right\} + \text{const.},
\end{equation}
where ${\bf P}=\bSigma^*-\bGamma^H\bSigma^{-1}\bGamma$, and ${\bf M}=\bGamma^H\bSigma^{-1}$ \cite{picinbono1996}; ${\bf P}^{-*}$ is a short notation for matrix inverse and conjugate value. Note that since $\bSigma=\bSigma^H$ and $\bGamma=\bGamma^T$, it holds that ${\bf P}={\bf P}^H$.

Let $\bpsi({\bf s}|\bSigma,\bGamma)$ denote the vector score function of ${\bf s}$, whose $k$th element is the $k$th score function of ${\bf s}$, i.e., $\psi_k({\bf s})={\bf e}_k^H\bpsi({\bf s}|\bSigma,\bGamma)$; ${\bf e}_k$ is the $k$th column of ${\bf I}_K$. By definition, it holds that 
\begin{multline}\label{eq:gaussscore}
  \bpsi({\bf s}|\bSigma,\bGamma) =-\frac{\partial \log f({\bf s})}{\partial {\bf s}}=\\{\bf P}^{-1}{\bf s}^* -\frac{1}{2}\bigl({\bf M}^T{\bf P}^{-*}+{\bf P}^{-1}{\bf M}\bigr){\bf s}.
\end{multline}
Let us consider the following well-known special cases. 
\begin{itemize}
    \item For the circular case, $\bGamma={\bf M}={\bf 0}$ holds and the vector score function takes on a simple form $\bpsi({\bf s}) = (\bSigma^{-1}{\bf s})^*$. 
    \item Let us consider the scalar case $K=1$, $\bSigma=\sigma^2=1$, and let us denote $\delta=\bGamma$. It then holds that $|\delta|\leq 1$. The score function takes on the form
    \begin{equation}\label{eq:gaussscoreICE}
        \psi(s)=\frac{1}{1-|\delta|^2}(s^*-\delta^*s).
    \end{equation}
\end{itemize}
The following Lemma will be useful for incorporating the Gaussian source model into the algorithms presented in the previous section.
\begin{lemma}\label{lemma5}
Let ${\bf s}$ be the $K$-dimensional Gaussian vector random variable with zero mean, covariance $\bSigma$, pseudo-covariance $\bGamma$, and score function $\bpsi({\bf s}|\bSigma,\bGamma)$. By the transformation theorem, $\bphi({\bf s})=\bpsi({\bf s}|\bLambda\bSigma\bLambda,\bLambda\bGamma\bLambda)$ is the score function of the normalized variable $\bLambda{\bf s}$ where $\bLambda={\tt diag}[\sigma_1^{-1},\dots,\sigma_K^{-1}]$; ${\tt diag}(\cdot)$ denotes the diagonal matrix with the values of the argument on its main diagonal. By definitions of \eqref{eq:nuktl} and \eqref{eq:rhoktl}, it holds that, for $k=1,\dots,K$,
\begin{align}
    \nu_{k}&=1,\label{eq:nugauss}\\
    \rho_{k}&=\sigma^2_{k}({\bf P}^{-1})_{kk}\label{eq:rhogauss}.
\end{align}
Next, when $\bSigma$, $\bGamma$, and $\bLambda$ are estimated, respectively, by $\widetilde\bSigma$, $\widetilde\bGamma$, and $\widetilde\bLambda$, and $\bphi({\bf s})=\bpsi({\bf s}|\widetilde\bLambda\widetilde\bSigma\widetilde\bLambda,\widetilde\bLambda\widetilde\bGamma\widetilde\bLambda)$ is then used as the model score function for the available samples of ${\bf s}$ and for $k=1,\dots,K$,
\begin{align}
    \hat\nu_{k}&=1,\label{eq:nugaussest}\\
    \hat\rho_{k}&=\hat\sigma^2_{k}(\widetilde{\bf P}^{-1})_{kk}\label{eq:rhogaussest},
\end{align}
where $\widetilde{\bf P}=\widetilde\bSigma^*-\widetilde\bGamma^H\widetilde\bSigma^{-1}\widetilde\bGamma$.
\end{lemma}
\begin{proof}
See Appendix~C.
\end{proof}

FastDIVA and QuickIVE based on the Gaussian source model are obtained when \eqref{eq:gaussscore}, \eqref{eq:nugaussest}, and \eqref{eq:rhogaussest} are put into \eqref{eq:fullgradient}, \eqref{eq:Hfastdivaktl}, and \eqref{eq:Hquickivektl}. The following three subsections consider particular variants of these algorithms.

\subsection{Scalar Gaussian SOI}
Here, we consider the special case corresponding to the fundamental static BSE problem when $K=T=1$ with the nonstationary Gaussian source model, i.e., $L>1$. We discuss the properties of stationary points of the contrast function, the simplified learning rules of FastDIVA and QuickIVE, and compare the circular and non-circular cases.

Let us consider the circular case first, where the score function is $\psi(s)=s^*$; it follows from \eqref{eq:gaussscoreICE} when $\delta=0$. The gradient \eqref{eq:normalizedgrad} is then obtained in the form
\begin{equation}\label{eq:normalizedgradGauss}
    \nabla={\bf a}-\left<\frac{\hat{\rm E}\left[\hat{s}_\ell^*{\bf x}_\ell\right]}{\hat\sigma_\ell^2}\right>_\ell={\bf a}-\left<{\bf a}_\ell\right>_\ell,
\end{equation}
where we used the definitions of \eqref{eq:OGC} and \eqref{eq:OGCsubblock}.
By putting \eqref{eq:normalizedgradGauss} equal to zero, and using that $\hat\sigma_\ell^2={\bf w}^H\widehat{\bf C}_\ell{\bf w}$, we obtain an elegant form of the condition for the stationary point of the contrast function \eqref{eq:contastICE}
\begin{equation}
    \frac{\bigl<\widehat{\bf C}_\ell\bigr>_\ell{\bf w}}{\bigl<{\bf w}^H\widehat{\bf C}_\ell{\bf w}\bigr>_\ell}=\left< \frac{\widehat{\bf C}_\ell{\bf w}}{{\bf w}^H\widehat{\bf C}_\ell{\bf w}}\right>_\ell.
\end{equation}
By considering $N\rightarrow+\infty$, it is seen that the SOI cannot be extracted when its variance $\sigma_\ell^2$ is constant over $\ell$, because any ${\bf w}$ satisfies this condition. This observation is in agreement with the identifiability condition of the corresponding BSE problem \cite{pham2001,Koldovsky2009}.

In the circular case, the learning rule of FastDIVA resp. QuickIVE is simplified to
\begin{equation}\label{eq:circularupdate}
    {\bf w}^{\rm new}={\bf w} - {\bf H}^{-1}\left({\bf a}-\left<{\bf a}_\ell\right>_\ell\right),
\end{equation}
where
\begin{equation}\label{eq:circularFastDIVA}
{\bf H}=\frac{\bigl<\widehat{\bf C}_\ell\bigr>_\ell}{\bigl<\hat\sigma_\ell^2\bigr>_\ell}-\left<\frac{\widehat{\bf C}_\ell}{\hat\sigma_\ell^2}\right>_\ell\quad\text{resp.}\quad
{\bf H}=-\left<\frac{\widehat{\bf C}_\ell}{\hat\sigma_\ell^2}\right>_\ell.
\end{equation}
It is seen that, when $\hat\sigma_\ell^2$ tends to be constant over $\ell$, the Hessian matrix of FastDIVA is close to zero, so the algorithm will have unstable behavior in the vicinity of the SOI. By taking  into account non-circularity, the above-mentioned learning rule is changed to
\begin{equation}\label{eq:noncircularupdate}
    {\bf w}^{\rm new}={\bf w} - {\bf H}^{-1}\left({\bf a}-\left<\frac{1}{1-|\hat\delta_\ell|^2}\left({\bf a}_\ell-\hat\delta_\ell^*\frac{\widehat{\bf D}_\ell{\bf w}^*}{\hat\sigma^2_\ell}\right)\right>_\ell\right),
\end{equation}
where $\widehat{\bf D}_\ell=\hat{\rm E}[{\bf x}_\ell{\bf x}_\ell^T]$ is the sample pseudo-covariance matrix of ${\bf x}_\ell$, and $\hat\delta_\ell={\bf w}^H\widehat{\bf D}_\ell{\bf w}^*/\hat\sigma^2_\ell$ is the normalized circularity coefficient of the SOI satisfying $|\hat\delta_\ell|\leq 1$, and 
\begin{equation}
{\bf H}=\frac{\bigl<\widehat{\bf C}_\ell\bigr>_\ell}{\bigl<\hat\sigma_\ell^2\bigr>_\ell}-\left<\frac{\widehat{\bf C}_\ell}{(1-|\hat\delta_\ell|^2)\hat\sigma_\ell^2}\right>_\ell, 
\end{equation}
and
\begin{equation}
{\bf H}=-\left<\frac{\widehat{\bf C}_\ell}{(1-|\hat\delta_\ell|^2)\hat\sigma_\ell^2}\right>_\ell, 
\end{equation}
for FastDIVA and QuickIVE, respectively.

We can see that, for $\hat\delta_\ell=0$, $\ell=1,\dots,L$, the algorithms given by \eqref{eq:circularupdate} and \eqref{eq:noncircularupdate} coincide. However, in experiments, we have observed that the latter update rule appears to be numerically more stable than the former one, because $\hat\delta_\ell$s are never exactly equal to zero even when the SOI is circular. 

The above-described algorithms can be easily extended to the nonstationary mixing conditions when $T>1$ by replacing the gradient and Hessian matrix by their averages over the blocks, as follows from the general formulas \eqref{eq:fullgradient} and \eqref{eq:Hfastdivaktl}.

\subsection{Vector Gaussian SOI: general covariance structure}\label{sec:gaussianPDF}

In the case of $K\geq 1$, the Gaussian source model can be directly applied to derive the update rules of FastDIVA and QuickIVE by considering  the Gaussian model score function given by \eqref{eq:gaussscore}. This approach is reasonable when no prior knowledge about $\bSigma$ or $\bGamma$ is available. The sample-based estimates $\widehat\bSigma=\hat{\rm E}[\hat{\bf s}\hat{\bf s}^H]$ and $\widehat\bGamma=\hat{\rm E}[\hat{\bf s}\hat{\bf s}^T]$ can be used where ${\bf s}$ stands for the current estimate of the SOI. We can then apply \eqref{eq:nugaussest} and \eqref{eq:rhogaussest} in Lemma~\ref{lemma5} with $\widetilde\bSigma=\widehat\bSigma$ and $\widetilde\bGamma=\widehat\bGamma$; subsequently we can replace the unknown matrices in \eqref{eq:rhogaussest} by their estimates, and put them into \eqref{eq:fullgradient}, \eqref{eq:Hfastdivaktl}, and \eqref{eq:Hquickivektl}.

However, two practical issues occur when $K$ gets "larger".
First, this approach tends to be stable and accurate until a critical number of samples within sub-blocks is available. The problem arises when $K\approx N_s$ because of the rank deficiency of $\widehat\bSigma$. Obviously, $\widehat\bSigma$ is singular when $K>N_s$. It can be avoided by adding a regularizing term to $\widehat\bSigma$ such as a multiple of the identity matrix, that is,
\begin{equation}\label{eq:diagonalloading}
 \widetilde\bSigma=\hat{\rm E}[\hat{\bf s}\hat{\bf s}^H] + \mu{\bf I}_K.
\end{equation}
Here, the parameter $\mu\geq 0$ provides a trade-off between the source modeling accuracy and algorithm stability. 

The second issue is that the computational complexity steeply grows with $K$ due to the computations of $\widetilde\bSigma^{-1}$ and $\widetilde{\bf P}^{-1}$. Although these matrices are Hermitian and positive definite,  the computations still take at least $\mathcal{O}(K^3)$ operations. This brings about a prohibitively large computational burden in some applications, such as in the frequency-domain audio source separation where $K$ corresponds to the frequency resolution (e.g., $K\geq 128$). Moreover, no significant algebraic simplifications in \eqref{eq:fullgradient} are possible unless these matrices have any favourable structure. Altogether, the approach discussed in this subsection is recommended only for ``small'' enough values of  $K$.

\subsection{Vector Gaussian SOI: tridiagonal covariance matrix}\label{sec:tridiagonal}
There are situations when $\widetilde\bSigma$ and $\widetilde\bGamma$ are structured, which can be used to alleviate the shortcomings of the previous approach. In this subsection, we consider the special case in which $\widetilde\bSigma$ is tridiagonal and $\widetilde\bGamma={\bf 0}$ (non-circularity is not taken into account). Without any loss of generality, we will consider $\widetilde\bSigma$ when all SOI components are normalized to have unit sample variance, so the assumed structure is given by
\begin{equation}\label{eq:tridiagonalsigma}
   \widetilde\bSigma=
    \begin{pmatrix}
       1 & c_1\\
       c_1^* & 1 & c_2\\
       &c_2^* & \ddots & \ddots\\
       &  & \ddots & \ddots & c_{K-1}\\
       &  &  & c_{K-1}^* & 1
    \end{pmatrix}.
\end{equation}
This structure means that adjacent components of the SOI are correlated. It is motivated by the situation that appears, e.g., in speech extraction in the short-term Frequency domain (STFT). Fig.~\ref{fig:speechcovarianceexample} shows a typical sample covariance matrix of normalized STFT channels of speech. There are significant correlations of adjacent frequency bands caused by their overlap; the other correlations appear to be less significant.

\begin{figure}
    \centering
    \includegraphics[width=0.85\linewidth]{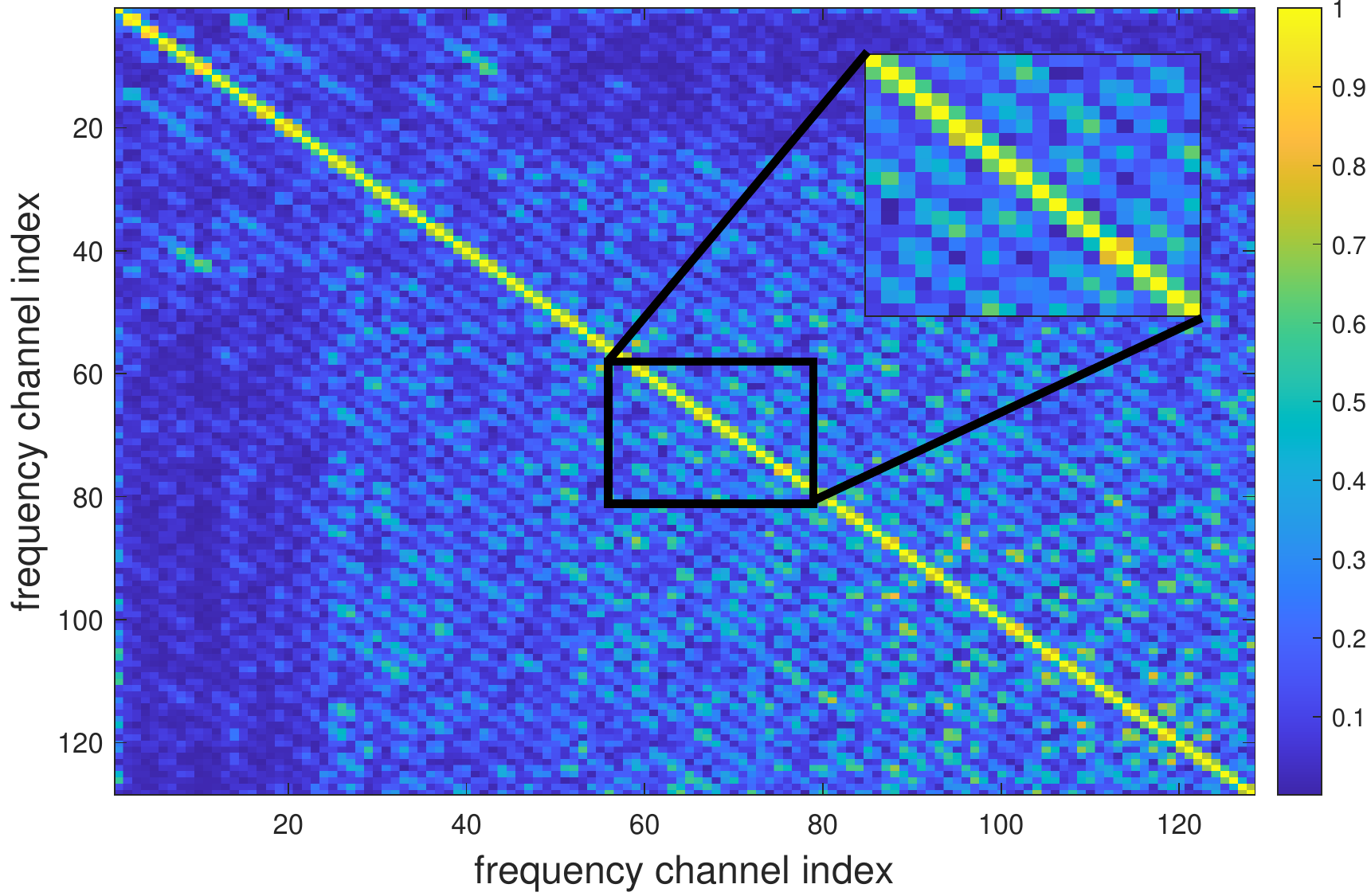}
    \caption{Example of a typical covariance matrix of normalized STFT channels of speech; a clear male speech sampled at $16$ kHz; FFT length of $256$ samples; window shift of $128$ samples; the Hamming analysis window used; average taken over $50$ frames.}
    \label{fig:speechcovarianceexample}
\end{figure}

By taking the advantage of this structure, the close-form formula from \cite{usmani1994} can now be used to compute $\widetilde\bSigma^{-1}$, which reduces the main computational burden needed for the evaluation of \eqref{eq:rhogaussest} and \eqref{eq:gaussscore}. We have that
\begin{equation}\label{eq:tridiagonalinverse}
  (\widetilde\bSigma^{-1})_{ij}=
  \begin{cases}
  (-1)^{i+j}c_i\dots c_{j-1}\theta_{i-1}\xi_{j+1}/\theta_K & i<j\\
   \theta_{i-1}\xi_{j+1}/\theta_K & i=j\\
  (-1)^{i+j}c_j^*\dots c_{i-1}^*\theta_{j-1}\xi_{i+1}/\theta_K & j<i\\
  \end{cases}
\end{equation}
for $i,j=1,\dots,K$, and
\begin{equation}
    \theta_i=\theta_{i-1}-|c_{i-1}|^2\theta_{i-2}, \quad    i=2,3,\dots,K,
\end{equation}
with initial conditions $\theta_0=\theta_1=1$, and
\begin{equation}
    \xi_i=\xi_{i+1}-|c_{i}|^2\xi_{i+2}, \quad    i=K-1,\dots,1,
\end{equation}
with initial conditions $\xi_{K+1}=\xi_K=1$. Using this, the computational burden due to $\widetilde\bSigma^{-1}$ is substantially reduced to almost linear complexity as follows. 

Since $|c_i|<1$, \eqref{eq:tridiagonalinverse} means that the off-diagonal entries of $\widetilde\bSigma^{-1}$ are exponentially decreasing with the growing distance from the main diagonal. We can therefore neglect the elements of $\widetilde\bSigma^{-1}$ on the $k$th diagonal for $|k|>k_{\rm max}$. The evaluation of $\widetilde\bSigma^{-1}$ then only takes  $\mathcal{O}(k_{\rm max}K)$ operations. Also, the multiplication by $\widetilde{\bf P}^{-1}$ in \eqref{eq:gaussscore}, which otherwise costs $\mathcal{O}(K^2)$, is reduced to $\mathcal{O}(k_{\rm max}K)$.

Another issue is that the positive definiteness of $\widetilde\bSigma$ must be ensured for stability of the algorithms. Here, we propose to constrain the off-diagonal entries of $\widetilde\bSigma$ as
\begin{equation}
    c_k=\begin{cases}
    \hat{c}_k & |\hat{c}_k|\leq 0.4\\
    0.4\cdot\frac{\hat{c}_k}{|\hat{c}_k|} & |\hat{c}_k|> 0.4
    \end{cases},
\end{equation}
where $\hat{c}_k=\hat{\rm E}[\hat s_k \hat s_{k+1}^*]$; $k=1,\dots,K-1$. The threshold for limiting the magnitude of $\hat{c}_k$ by $0.4$ is inspired by the analytic value of the eigenvalues of tridiagonal matrices where $c_k$s are all constant and equal to $c$. Their eigenvalues are
\begin{equation}
    1+2|c|\cos\left(\frac{k\pi}{K+1}\right), \quad k=1,\dots,K.
\end{equation}
In that case, the limit $0.4$ hence ensures that all eigenvalues of $\widetilde\bSigma$ are sufficiently larger than zero.

\section{Experimental Validation}

\subsection{Simulations}
The proposed algorithms are validated in simulated experiments and compared with other state-of-the-art algorithms. Their performance is assessed in terms of the interference-to-signal ratio (ISR) measured on the extracted signal(s). The 1\% trimmed mean is used for averaging over Monte Carlo repetitions in order to avoid trials where the given algorithm extracts a different independent source than the SOI. Owing to the ambiguity of the BSE task, these cases do not necessarily mean failures.

In a simulation trial, the SOI(s) samples are drawn independently according to the complex Generalized Gaussian law \cite{cGGD} with zero mean, (normalized) circularity coefficient $|\delta|\leq 1$, and the shape parameter $c>0$ ($c=1$ means Gaussian, $c<1$ super-Gaussian, and $c>1$ sub-Gaussian). The nonstationarity of the SOI is driven through its variance, which is, on the $t$th block and $\ell$th subblock, equal to
\begin{equation}\label{eq:SOIvariance}
    \sigma^2_{t,\ell}=\left[\sin\left(\frac{t\pi}{T+1}\right)\sin\left(\frac{\ell\pi}{L+1}\right)\right]^\alpha.
\end{equation}
The SOI is stationary for $\alpha=0$, moderately dynamical for $\alpha\approx 1$, and transient-like when $\alpha\gg 1$; see the example shown in Fig.~\ref{fig:alpha}. The background sources are generated as stationary uncorrelated circular Gaussian.

\begin{figure}[ht]
    \centering
    \includegraphics[width=0.9\linewidth]{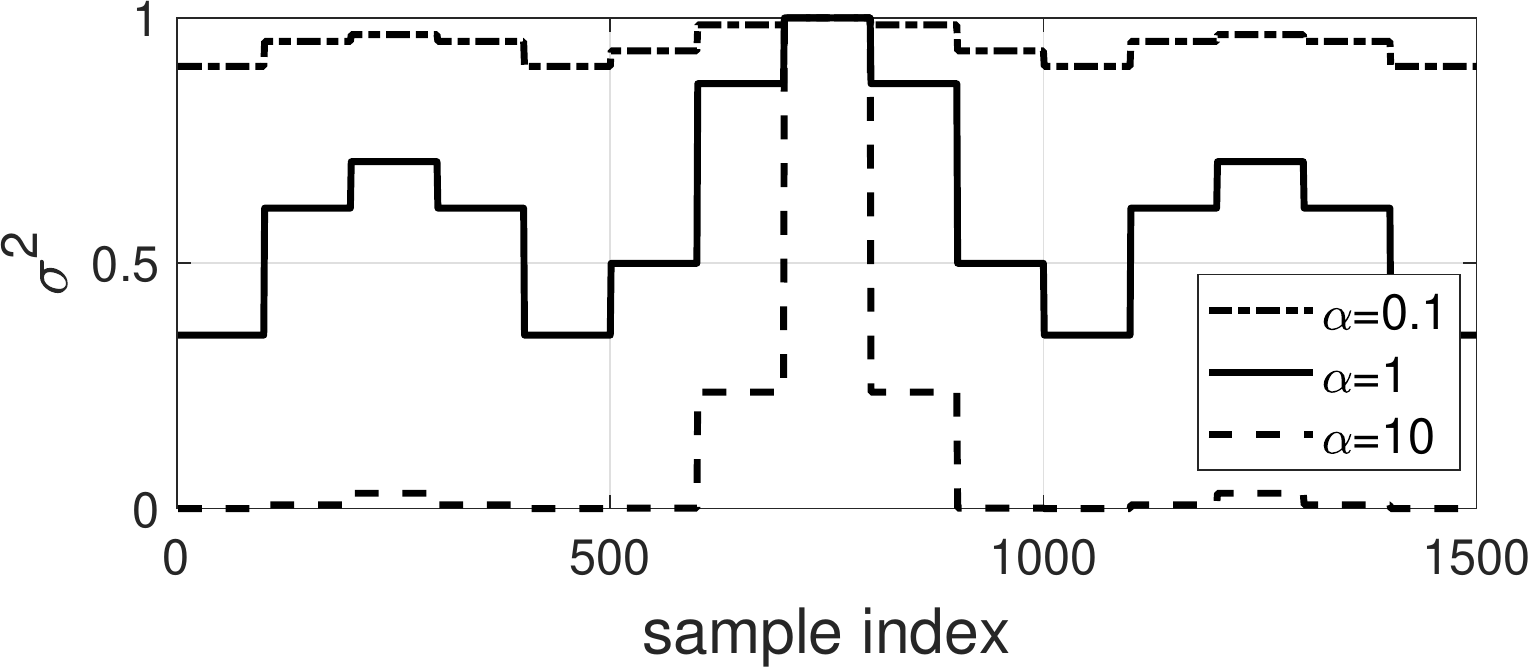}
    \caption{Example of variance profiles of the SOI according to \eqref{eq:SOIvariance} when $T=3$, $L=5$, $N_s=100$.}
    \label{fig:alpha}
\end{figure}

The signals are mixed, in each block, by a random mixing matrix such that that the CSV mixing model is obeyed.  
The initial signal-to-interference ratio on input channels is chosen to be approximately constant \cite{koldovsky2019TSP}. The algorithms are initialized by the separating vectors ${\bf w}^{\rm ini}_k={\bf w}_k^\star+\boldsymbol\epsilon_k$ where $\boldsymbol\epsilon_k$ is a random vector orthogonal to ${\bf w}_k^\star$ such that $\|\boldsymbol\epsilon_k\|^2=0.01$. 

\subsubsection{Static Independent Component Extraction}
The standard ICE problem with $T=K=1$ is considered here where $L=20$, $d=6$, and $N=5000$, i.e., $N_s=250$. The SOI distribution is generated with $c=1$ and $\delta=0.5$ (non-circular Gaussian); the case when $c=0.5$ (non-circular Laplacean) is provided in the supplementary material of this article. 

We have compared eight algorithms. BOGIVE$_{\bf w}$ \cite{koldovsky2019icassp,jansky2022}, FastDIVA from \cite{koldovsky2021fastdiva}, and CSV-AuxIVE \cite{jansky2022} represent non-Gaussianity-based methods that inherently work with the hypothesis that $L=1$. In BOGIVE$_{\bf w}$ and FastDIVA, the rational nonlinearity $\phi(s)=\frac{s^*}{1+|s|^2}$ denoted as ``rati'' is used; CSV-AuxIVE utilizes the standard nonlinearity for super-Gaussian sources \cite{ono2011stable}. As for a method assuming Gaussianity and non-circularity and allowing for nonstationarity, we compare LLJBD from \cite{tichavsky2009}. The extended versions of FastDIVA and QuickIVE proposed in this article stand for the methods allowing for the non-Gaussinity and/or nonstationarity. They are tested with $L=20$ and with the nonlinearities ``rati'' or ``gauss'', where the latter  corresponds to \eqref{eq:gaussscoreICE}, i.e., the Gaussian score allowing for non-circularity. From here, short abbreviations {\em algorithm--nonlinearity-$L$} will be used for the variants of FastDIVA and QuickIVE, e.g.,  FastDIVA--rati--$20$.

Fig.~\ref{fig:ICE_gaussian} shows the results when $c=1$ for $1,000$ trials. For $\alpha=0.1$, the algorithms yield poor ISR ($>-15$dB) since the SOI is almost stationary and Gaussian, up to FastDIVA--gauss--$20$ and QuickIVE--gauss--$20$ that benefit from the non-circularity of the SOI. With growing $\alpha$, all algorithms are improving as the SOI becomes more nonstationary, including the non-Gaussianity-based FastDIVA--rati--$1$ and BOGIVE$_{\bf w}$. However, BOGIVE$_{\bf w}$ fails to achieve the same performance as FastDIVA due to its slow convergence (it is limited by the maximum number of $1,000$ iterations). 

The performance is better for algorithms capturing the nonstationarity considering $L=20$; especially, for $\alpha>1$. The superior performance is achieved by FastDIVA--gauss--$20$ and QuickIVE--gauss--$20$, whose performance characteristics coincide in this case. The superiority is achieved due to the accurate source modelling including non-circularity. These methods yield lower performance with the ``rati'' nonlinearity (assuming circularity); FastDIVA--rati--$20$ shows less stable convergence than QuickIVE--rati--$20$ for the values of $\alpha>1$. 

LLJBD is also improving with growing $\alpha$; however, the median ISR has to be shown here due to unstable convergence. Surprising results are obtained by CSV-AuxIVE since it shows improvement with growing  $\alpha$ similarly to the methods employing nonstationarity. This is in contrast with the fact that the method comes from the optimization of  non-Gaussianity-based source model \cite{ono2011stable,jansky2022}. The theoretical explanation of this behavior goes beyond the scope this paper.

\begin{figure}
    \centering
    \includegraphics[width=\linewidth]{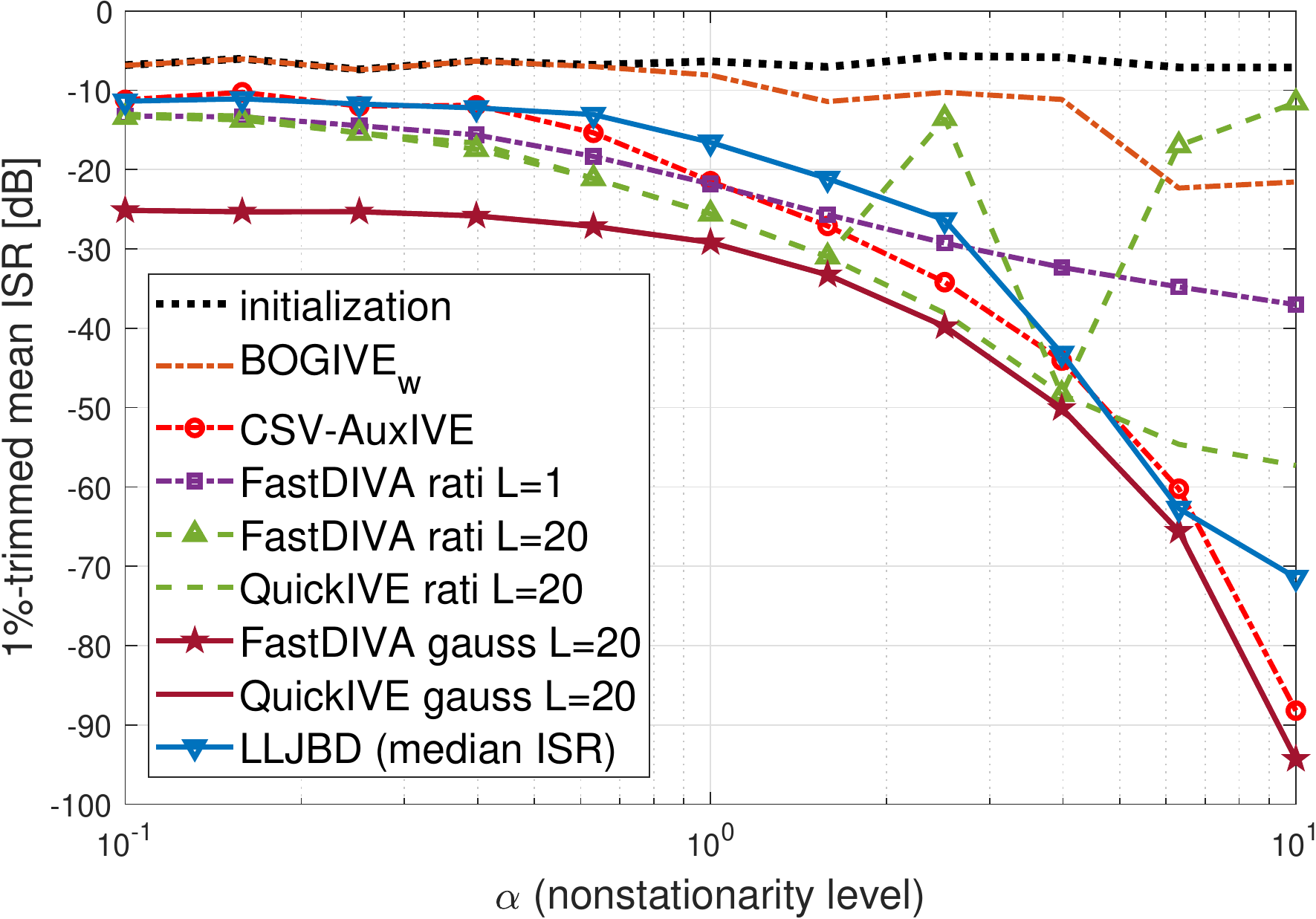}
    \caption{Resulting ISR as a function of $\alpha$, the parameter controlling the nonstationarity of the SOI according to \eqref{eq:SOIvariance}. The pdf of the SOI is Gaussian as $c=1$ with circularity coefficient $\delta=0.5$; ``initialization'' corresponds to the do-nothing approach and reflects the ISR given by the initialization.}
    \label{fig:ICE_gaussian}
\end{figure}

\subsubsection{Dynamic Independent Component Extraction}

\begin{figure}
    \centering
    \includegraphics[width=\linewidth]{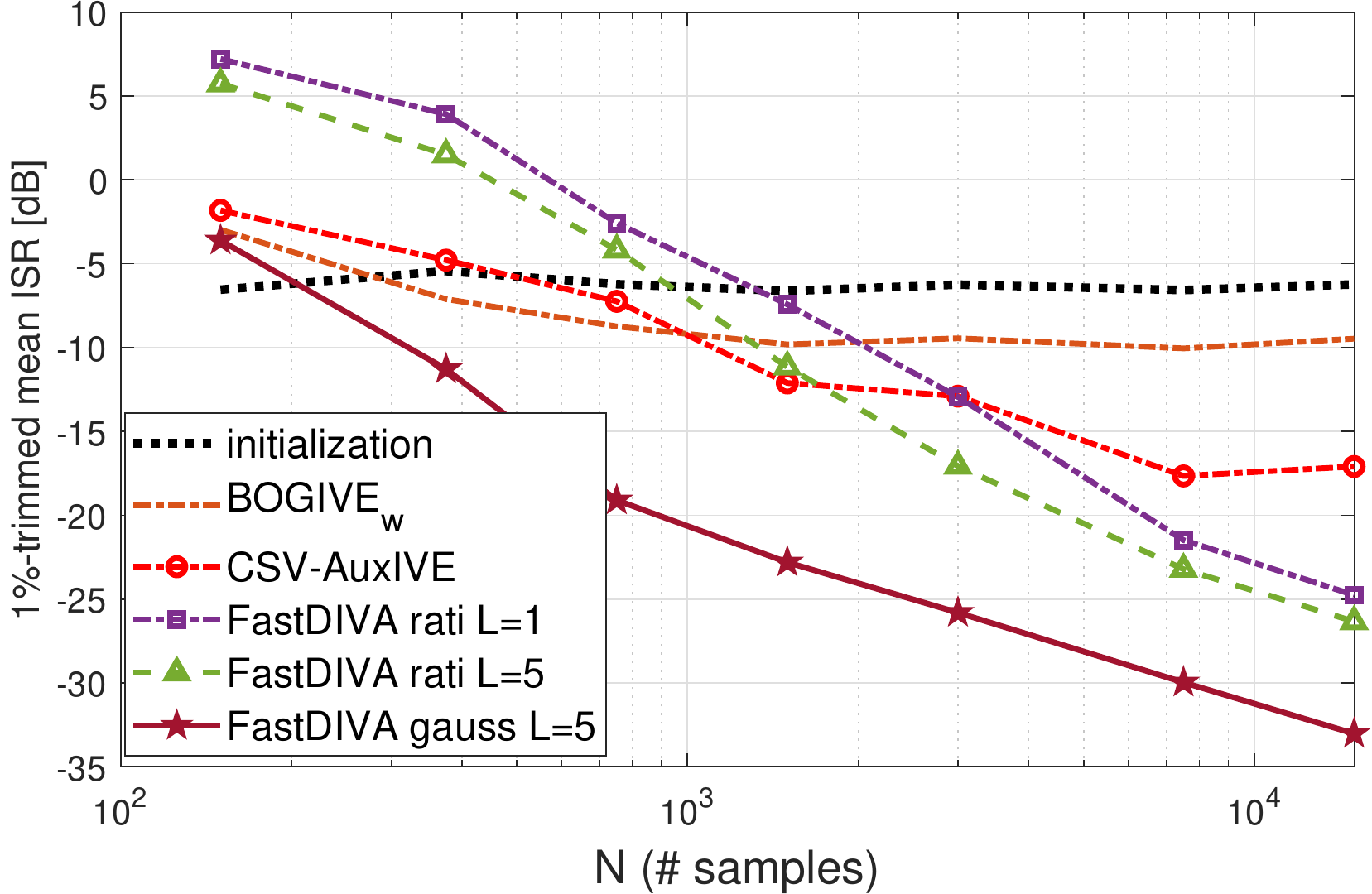}
    \caption{ISR as a function of $N$ when $c=1$ and $\delta=0.5$ (Gaussian non-circular SOI), $T=3$ (nonstationary CSV mixing model), and $L=5$ and $\alpha=2$ (nonstationary source model). $N$ ranges from $150$ through $15,000$; the length of sub-blocks $N_s/(TL)$ thus ranges from $15$ through $1,000$.}
    \label{fig:CSV_gaussian}
\end{figure}

We now turn to  $T=3$, $K=1$, $d=6$, and where the SOI is nonstationary and non-circular Gaussian with $L=5$, $\alpha=2$ and $\delta=0.5$. Since $T>1$, we compare only the methods that allow for the CSV mixing model. QuickIVE is not presented here since it has provided the same results as FastDIVA.

Fig.~\ref{fig:CSV_gaussian} shows the average ISR achieved after $1,000$ trials as a function of the length of data $N$. The values of $N$ are selected so that $N_s$ ranges from the extremely small value of $N_s=15$ through $N_s=1,000$. With growing $N$, all methods are improving, including BOGIVE$_{\bf w}$, CSV-AuxIVE, and FastDIVA--rati--$1$, which do not exploit the nonstationarity of the SOI on the sub-blocks. The extended FastDIVA (allowing for $L=5$) shows better performance than with $L=1$. Moreover, FastDIVA--gauss--20 takes the advantage of  non-circularity and achieves a useful accuracy level (of $\approx -11$~dB) even in the extreme case of $N_s=25$ (resp. $N=150$).
The experiment when the SOI is Laplacean is shown in the supplementary material.

\subsubsection{Independent Vector Extraction}
Here, we consider $K=5$ mixtures of dimension $d=10$ involving jointly dependent components of the SOI. These components are generated as follows: First, five signals are generated independently with $\alpha=2$, $c=0.5$, $\delta=0.5$, $L=10$ (nonstationary Laplacean non-circular sources). Second, these signals are multiplied by a random $K\times K$ matrix drawn from $\mathcal{CN}(0,1)$, which yields dependent and correlated SOI components. $K$ mixtures with $T=1$ are generated, so the data obeys the standard IVE mixing model. The number of samples is $N=500$; $N_s$ is $50$.
The compared algorithms are tested in two regimes: the $K$ mixtures are processed (ICE) separately and (IVE) jointly.

Fig.~\ref{fig:IVE_laplacean_speed} shows ISR averaged over the mixtures and $1,000$ trials as a function of iteration index. 
In the ICE regime, these methods show significantly slower convergence than in the IVE regime (e.g., CSV-AuxIVE or FastDIVA--rati--$1$) and also lower accuracy levels because the dependencies among the SOI components are not used. This shows that proper source modeling influences not only the algorithms' accuracy but also their convergence. FastDIVA is shown to be mostly the fastest method among the compared ones; QuickIVE provides its slightly decelerated (and sometimes more stable) variant.

\begin{figure}
    \centering
    \includegraphics[width=\linewidth]{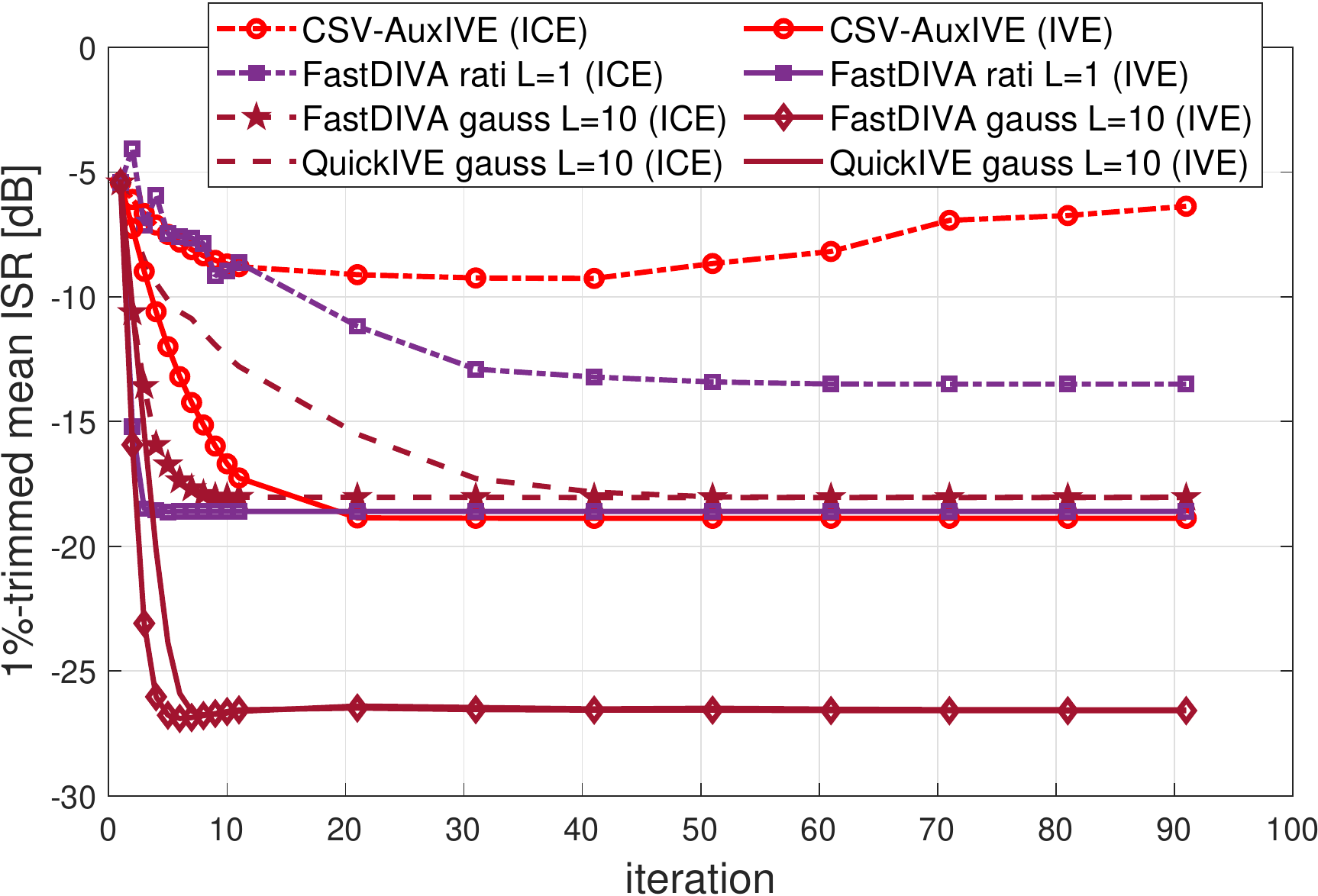}
    \caption{Mean convergence of algorithms performing separate (ICE) and joint (IVE) blind extraction: ISR as a function of iteration index. The parameters of the experiment are $\alpha=2$, $c=0.5$, $\delta=0.5$, $L=10$ (nonstationary Laplacean non-circular SOI), $d=10$, $T=1$ (static mixtures),  $K=5$ (five jointly dependent mixtures per trial), and $N=500$.}
    \label{fig:IVE_laplacean_speed}
\end{figure}

\subsection{Frequency-domain Blind Speech Extraction}

The typical application of IVE includes speech extraction in the short time frequency domain. The mixture of speech and background are convolutive in the time domain and can be approximated as instantaneous in the frequency domain; hence the instantaneous complex-valued ICE and IVE mixing models can be applied. Compared to ICE, IVE tries to secure that the SOI is extracted in each frequency band (the permutation problem) by using dependencies \cite{kim2007}.
However, it is generally known that this solution through IVE is not definite. For example, the extracted frequency components can form groups corresponding to different independent sources, so, finally, the complete extraction is not achieved. The experiment proposed here aims at a detailed investigation of the convergence issues of the compared algorithms in this application. Also, the variants of FastDIVA/QuickIVE for Gaussian SOI with tridiagonal covariance matrix proposed in Section~\ref{sec:tridiagonal} are employed here (denoted with the ``gausstri'' nonlinearity).

In a trial, a short interval of clean speech is transformed by the Short-Time Fourier Transform (STFT) with the FFT and shift length, respectively, equal to $2K+1$ and $K$ samples; we consider $K=128$; the number of the STFT frames corresponds to $N=375$. 
Then, $K$ mixtures, one per each frequency band $2,\dots,K+1$, obeying the CSV model with $d=10$ (simulating $10$ microphones) and $T=3$ (a moving speaker) are generated. The speech frequency components play the role of the SOI components; the background is generated from $d-1$ independent Laplacean signals. 

The results of five algorithms are evaluated in Fig.~\ref{fig:CSV_speech}, where the ISR is shown as a function of iteration index. The graphs in the first row show median ISR taken over $100$ trials; the lines show the average ISR taken over all $K$ frequencies while the transparent areas show the range from the minimum through the maximum ISR over the frequencies. For more insight, the charts in the second row illustrate the results of trial 1.

The example shows that all algorithms tend to extract most of the frequency components of the speech. This is indicated by the average median ISR whose value goes significantly below $0$~dB. However, the ISR of some frequencies is sometimes growing, which seems to happen more often for CSV-AuxIVE and FastDIVA--rati--$1$. FastDIVA--gauss--$5$ (Section~\ref{sec:gaussianPDF}) appears to yield a more stable convergence in all frequencies; however, it is significantly slower than the other methods and is computationally very expensive due to $K=128$. The fastest and most reliable convergence is observed in FastDIVA--gausstri--$5$ and  QuickIVE--gausstri--$5$. The average median ISR by these methods achieves values below $-20$~dB after less than $10$ iterations, which is the superior extraction accuracy among the compared methods.

\begin{figure*}
    \centering
    \includegraphics[width=\linewidth]{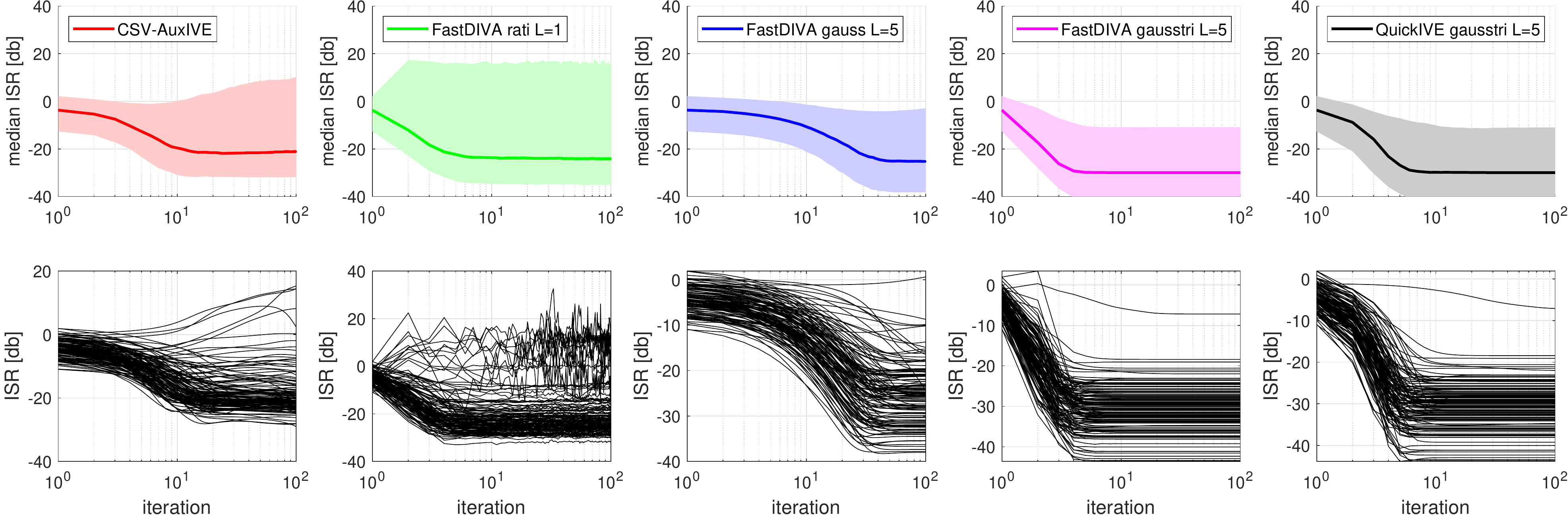}
    \caption{Results of the simulated frequency-domain speech extraction. Row~1: median ISR taken over $100$ trials; the average ISR (line); the range from the minimum through the maximum value of ISR over all frequencies (area). Row~2: the ISR of all frequencies in trial 1.}
    \label{fig:CSV_speech}
\end{figure*}

\section{Conclusions}

The BSE model, combining non-Gaussianity and nonstationarity in the source model and nonstationarity in the mixing models, makes the set of identifiable sources broader. We have derived extended variants of FastDIVA and QuickIVE, which show faster convergence and higher accuracy than the state-of-the-art methods. We have shown that these algorithms can be used with the Gaussian score function, which makes them purely based on  second-order statistics. In complex-valued problems, they can efficiently exploit non-circularity. The special variant, assuming jointly Gaussian SOIs with tridiagonal covariance matrix, shows promising results for  frequency-domain blind speaker extraction.

\section*{Appendix A: Proof of Lemma 1}
The gradient of the first two terms in \eqref{eq:contastICE}, here denoted as $\nabla_{12}$, readily gives that
\begin{multline}\label{eq:gradient_terms2}
    \nabla_{12}=\frac{\partial}{\partial {\bf w}^H} \Bigg<\hat{\rm E}\left[\log f\left(\frac{\hat{s}_\ell}{\hat\sigma_\ell}\right)\right]  -\log\hat\sigma_\ell^2\Bigg>_\ell=\\
    -\Bigg<\hat{\rm E}\left[\phi\left(\frac{\hat{s}_\ell}{\hat\sigma_\ell}\right)\frac{{\bf x}_\ell}{\hat\sigma_\ell}\right]+\Re(\hat\nu_\ell){\bf a}_\ell - {\bf a}_\ell\Bigg>_\ell,
\end{multline}
where ${\bf a}_\ell$ is defined by \eqref{eq:OGCsubblock}, and where we used identities
\begin{align}
    \frac{\partial}{\partial {\bf w}^H}{\hat s}_\ell&=\frac{\partial}{\partial {\bf w}^H}{\bf w}^H{\bf x}_\ell={\bf x}_\ell,\\
    \frac{\partial}{\partial {\bf w}^H}\frac{1}{\hat\sigma_\ell}&=
    \frac{\partial}{\partial {\bf w}^H}\frac{1}{\sqrt{{\bf w}^H\widehat{\bf C}_\ell{\bf w}}}=-\frac{\bf a}{2\hat\sigma_\ell},\\
    \frac{\partial}{\partial {\bf w}^H}\log\hat\sigma^2_\ell&= \frac{\partial}{\partial {\bf w}^H}\log{\bf w}^H\widehat{\bf C}_\ell{\bf w}={\bf a}_\ell.
\end{align}

We continue by showing that the gradient of the last two terms in \eqref{eq:contastICE}, denoted as $\nabla_{34}$, reads
\begin{equation}\label{eq:gradient_terms1}
    \nabla_{34}=\frac{\partial}{\partial {\bf w}^H}\left<-\hat{\rm E}\left[\hat{\bf z}_\ell^H{\bf R}_\ell\hat{\bf z}_\ell\right] + (d-2)\log |\gamma|^2\right>_\ell = {\bf a}.
\end{equation}
To this end, note that the two terms inside the averaging operator $\left<\cdot\right>_\ell$ in \eqref{eq:contastICE} have the same analytic shape as the last two terms in Eq. 26 in \cite{koldovsky2019TSP}. Therefore, we can employ the result given by Eq. 27 in \cite{koldovsky2019TSP}, taking into account the linearity of the operator $\left<\cdot\right>_\ell$ and the dependency of the signals and their statistics on $\ell$. By considering substitutions ${\bf z}\rightarrow{\bf z}_\ell$, $\widehat{\bf C}_{\bf z}\rightarrow\widehat{\bf C}_{\bf z}^\ell$, ${\bf R}\rightarrow{\bf R}_\ell$, we obtain that
\begin{multline}\label{lemma1:aux1}
  \nabla_{34}=  2\a\,\tr({\bf R}_\ell\bigl<\widehat{\bf C}_{\bf z}^\ell\bigr>_\ell)\\
-({\bf w}^H\widehat{\bf C}{\bf w})^{-1}\bigl(\widehat{\bf C}{\bf E}^H{\bf R}_\ell\bigl<\widehat{\bf C}_{\bf z}^\ell\bigr>_\ell\h-\tr({\bf R}_\ell{\bf B}\widehat{\bf C}{\bf 
E}^H)\widehat{\bf C}{\bf e}_1\bigr)\\
-2(d-2)\a+(\gamma^*)^{-1}(d-2)({\bf w}^H\widehat{\bf C}{\bf w})^{-1}\widehat{\bf C}{\bf e}_1,
\end{multline}
where $\tr(\cdot)$ denotes the trace; ${\bf E}=[{\bf 
0}\quad\I_{d-1}]$; ${\bf e}_1$ denotes the first column of $\I_d$. 
Now, we put ${\bf R}_\ell=\bigl<\widehat{\bf C}_{\bf z}^\ell\bigr>_\ell^{-1}$ as assumed by the Lemma, \eqref{lemma1:aux1} turns to
\begin{multline}\label{lemma1:aux2}
\nabla_{34}=  2\a
-({\bf w}^H\widehat{\bf C}{\bf w})^{-1}\bigl(\widehat{\bf C}{\bf E}^H\h\\-\tr(\bigl<\widehat{\bf C}_{\bf z}^\ell\bigr>_\ell^{-1}{\bf B}\widehat{\bf C}{\bf 
E}^H)\widehat{\bf C}{\bf e}_1\bigr)\\
+(\gamma^*)^{-1}(d-2)({\bf w}^H\widehat{\bf C}{\bf w})^{-1}\widehat{\bf C}{\bf e}_1.
\end{multline}
By similar steps to those given by Eq. 31 and 32 in \cite{koldovsky2019TSP}, 
\begin{multline}\label{lemma1:aux3}
    \bigl<\widehat{\bf C}_{\bf z}^\ell\bigr>_\ell^{-1}{\bf B}\widehat{\bf C}=\bigl<\widehat{\bf C}_{\bf z}^\ell\bigr>_\ell^{-1}{\bf B}\bigl<\hat{\rm E}[{\bf x}_\ell{\bf x}_\ell^H]\bigr>_\ell=\\ \bigl<\widehat{\bf C}_{\bf z}^\ell\bigr>_\ell^{-1}\bigl<\hat{\rm E}\bigl[{\bf z}_\ell[\hat{s}_\ell^*\quad\hat{\bf z}_\ell^H]\bigr]\bigr>_\ell{\bf A}^H=\\
    \bigl<\widehat{\bf C}_{\bf z}^\ell\bigr>_\ell^{-1}\bigl[{\bf 0}\quad \bigl<\widehat{\bf C}_{\bf z}^\ell\bigr>_\ell\bigr]{\bf A}^H={\bf E}{\bf A}^H,
\end{multline}
where we have used that $\hat{\rm E}[{\bf z}_\ell\hat{s}_\ell^*]={\bf 0}$, which follows from \eqref{eq:OGC_over_subblocks}. Hence,
\begin{equation}\label{lemma1:aux4}
  \tr(\bigl<\widehat{\bf C}_{\bf z}^\ell\bigr>_\ell^{-1}{\bf B}\widehat{\bf C}{\bf E}^H)=\tr({\bf E}{\bf A}^H{\bf E}^H)=-\beta-(d-2)(\gamma^*)^{-1}.
\end{equation}
By putting \eqref{lemma1:aux4} into \eqref{lemma1:aux3} and using \eqref{eq:OGC}, \eqref{lemma1:aux1} follows. The assertion of the Lemma follows by summing \eqref{eq:gradient_terms2} and \eqref{eq:gradient_terms1}.\hfill\rule{1.2ex}{1.2ex}

\section*{Appendix B: Proof of Lemma 3 and 4}
The proof follows analogous steps to those in Appendix~A in \cite{koldovsky2021fastdiva}.
For $N\rightarrow+\infty$, the sample-based estimates are replaced by the expectation values, and, by definition,
\begin{align}
{\bf H}_1&=
\frac{\partial \nabla^H}{\partial{\bf w}}=\left[\frac{\partial}{\partial{\bf w}^H}\left({\bf a}^T-\left<\nu_\ell^{-1}{\rm E}\left[\phi\left(\frac{s_\ell}{\sigma_\ell}\right)\frac{{\bf x}_\ell^T}{\sigma_\ell}\right]\right>_\ell\right)\right]^*\label{B:H1}\\
{\bf H}_2&=
\frac{\partial \nabla^T}{\partial{\bf w}}=\frac{\partial}{\partial{\bf w}}\left({\bf a}^T-\left<\nu_\ell^{-1}{\rm E}\left[\phi\left(\frac{s_\ell}{\sigma_\ell}\right)\frac{{\bf x}_\ell^T}{\sigma_\ell}\right]\right>_\ell\right)\label{B:H2}.
\end{align}
The dependent variables on ${\bf w}$ are $s_\ell={\bf w}^H{\bf x}_\ell$ and $\sigma_\ell^2={\bf w}^H{\bf C}_\ell{\bf w}$; $\nu_\ell$ are treated as constants. For proving both of the Lemmas, the 
following identities will be used:
\begin{align}
  \frac{\partial}{\partial{\bf w}}\frac{1}{\sigma_\ell}&=
 -\frac{{\bf a}_\ell^*}{2\sigma_\ell}, &
 \frac{\partial}{\partial{\bf w}^H}\frac{1}{\sigma_\ell}&=
 -\frac{{\bf a}_\ell}{2\sigma_\ell},\\
 \frac{\partial}{\partial{\bf w}}\frac{s_\ell^*}{\sigma_\ell}&=
 \frac{{\bf x}_\ell^*}{\sigma_\ell}-\frac{s_\ell^*{\bf a}_\ell^*}{2\sigma_\ell}, &
 \frac{\partial}{\partial{\bf w}^H}\frac{s_\ell}{\sigma_\ell}&=
 \frac{{\bf x}_\ell}{\sigma_\ell}-\frac{s_\ell{\bf a}_\ell}{2\sigma_\ell},
\end{align}
and
\begin{align}
&\frac{\partial}{\partial{\bf w}^H}\phi_\ell\frac{{\bf x}^T_\ell}{\sigma_\ell}=\left(
   \frac{\partial\phi_\ell}{\partial s_\ell}\left(\frac{{\bf x}_\ell}{\sigma_\ell}-\frac{s_\ell{\bf a}_\ell}{2\sigma_\ell}\right)-
   \frac{\partial\phi_\ell}{\partial s_\ell^*}
   \frac{s_\ell^*{\bf a}_\ell}{2\sigma_\ell}
   -\phi_\ell\frac{{\bf a}_\ell}{2}\right)\frac{{\bf x}_\ell^T}{\sigma_\ell},\nonumber\\
   &\frac{\partial}{\partial{\bf w}}\phi_\ell\frac{{\bf x}_\ell^T}{\sigma_\ell}=\left(
   \frac{\partial\phi_\ell}{\partial s_\ell^*}\left(\frac{{\bf x}_\ell^*}{\sigma_\ell}-\frac{s_\ell^*{\bf a}_\ell^*}{2\sigma_\ell}\right)-\frac{\partial\phi_\ell}{\partial s_\ell}\frac{s_\ell{\bf a}_\ell^*}{2\sigma_\ell}-\phi\frac{{\bf a}_\ell^*}{2}\right)\frac{{\bf x}_\ell^T}{\sigma_\ell}.\nonumber
\end{align}
where $\phi_\ell$ is a short notation of $\phi\left(\frac{s_\ell}{\sigma_\ell}\right)$. Considering the expectation values of the latter two expressions, and the fact that ${\bf x}_\ell={\bf a}s_\ell+{\bf y}_\ell$ where ${\bf s}_\ell$ and ${\bf y}_\ell$ are independent, we obtain
\begin{align}
\frac{\partial}{\partial{\bf w}^H}{\rm E}\left[\phi_\ell\frac{{\bf x}_\ell^T}{\sigma_\ell}\right]&=
    \left[\eta_\ell{\bf a}
    -\frac{\tau_\ell}{2}{\bf a}_\ell\right]{\bf a}^T+\frac{\pi_\ell}{\sigma_\ell^2}{\bf P}_{\bf y}^\ell,\label{B:exp4}\\
    \frac{\partial}{\partial{\bf w}}{\rm E}\left[\phi_\ell\frac{{\bf x}_\ell^T}{\sigma_\ell}\right]&=\left[\xi_\ell{\bf a}^*
    -\frac{\tau_\ell}{2}{\bf a}_\ell^*\right]{\bf a}^T+\frac{\rho_\ell}{\sigma_\ell^2}({\bf C}^\ell_{\bf y})^*\label{B:exp5},
\end{align}
where $\pi_\ell={\rm E}[\frac{\partial\phi_\ell}{\partial s}]$, and ${\bf P}^\ell_{\bf y}={\rm E}[{\bf y}_\ell{\bf y}_\ell^T]$ and ${\bf C}_{\bf y}^\ell={\rm E}[{\bf y}_\ell{\bf y}_\ell^H]$ are the covariance and pseudo-covariance of ${\bf y}_\ell$, respectively. Owing to the assumption stated in Section~\ref{sec:sourcemodel} that the background signals are circular Gaussian, ${\bf P}^\ell_{\bf y}={\bf 0}$.
Next, it holds that ${\bf C}_\ell={\bf a}{\bf a}^H\sigma_\ell^2+{\bf C}_{\bf y}^\ell$. By making these substitutions in \eqref{B:exp4} and \eqref{B:exp5}, we obtain
\begin{align}
\frac{\partial}{\partial{\bf w}^H}{\rm E}\left[\phi_\ell\frac{{\bf x}_\ell^T}{\sigma_\ell}\right]&=
    \left[\eta_\ell{\bf a}
    -\frac{\tau_\ell}{2}{\bf a}_\ell\right]{\bf a}^T,\label{B:exp6}\\
\frac{\partial}{\partial{\bf w}}{\rm E}\left[\phi_\ell\frac{{\bf x}^T_\ell}{\sigma_\ell}\right]&=\left[(\xi_\ell-\rho_\ell){\bf a}^*
    -\frac{\tau_\ell}{2}{\bf a}_\ell^*\right]{\bf a}^T+\frac{\rho_\ell}{\sigma_\ell^2}{\bf C}_\ell^*.\label{B:exp7}
\end{align}

By inserting \eqref{B:exp6} and \eqref{B:exp7}, respectively, into \eqref{B:H1} and \eqref{B:H2}, the assertion of Lemma~4 follows.

What is left to compute for the proof of Lemma~3 are the derivatives of ${\bf a}^T$ in \eqref{B:H1} and \eqref{B:H2} when the OGC \eqref{eq:OGC} is imposed:
\begin{align}
 \frac{\partial {\bf a}^T}{\partial{\bf w}^H}&=
 -{\bf a}{\bf a}^T,\label{B:exp8} \\
   \frac{\partial {\bf a}^T}{\partial{\bf w}}&=
 \frac{{\bf C}^*}{\sigma^2}-{\bf a}^*{\bf a}^T, \label{B:exp9}
 \end{align}
where $\sigma^2=\bigl<\sigma^2_\ell\bigr>_\ell$.
By inserting \eqref{B:exp6} and \eqref{B:exp8} into \eqref{B:H1}, and \eqref{B:exp7} and \eqref{B:exp8} into \eqref{B:H2}, the assertion of Lemma~3 follows.
\hfill\rule{1.2ex}{1.2ex}

\section*{Appendix C: Proof of Lemma~5}
By definition, ${\bf P}=\bSigma^*-\bGamma^H\bSigma^{-1}\bGamma$ and ${\bf M}=\bGamma^H\bSigma^{-1}$. Using the matrix inverse lemma, it can be shown that
\begin{equation}\label{AC:identity}
  {\bf P}^{-1}\bSigma^* -\frac{1}{2}\bigl({\bf M}^T{\bf P}^{-*}+{\bf P}^{-1}{\bf M}\bigr)\bGamma = {\bf I}_K.  
\end{equation}
We will use this identity for proving \eqref{eq:nugauss}.

By considering the substitutions $\bSigma\leftarrow\bLambda\bSigma\bLambda$ and $\bGamma\leftarrow\bLambda\bGamma\bLambda$ and \eqref{eq:gaussscore}, the score function of the normalized ${\bf s}$ reads
\begin{multline}\label{AC:modelscore}
      \bphi({\bf s})=\bpsi({\bf s}|\bLambda\bSigma\bLambda,\bLambda\bGamma\bLambda)=
      \bLambda^{-1}{\bf P}^{-1}\bLambda^{-1}{\bf s}^* \\-\frac{1}{2}\bLambda^{-1}\bigl({\bf M}^T{\bf P}^{-*}+{\bf P}^{-1}{\bf M}\bigr)\bLambda^{-1}{\bf s}.
\end{multline}
We can now use this formula to express the following matrix:
\begin{multline}\label{AC:auxmatrix}
    \boldsymbol{\Xi}={\rm E}[\bphi(\bLambda{\bf s}){\bf s}^T\bLambda]=
    \bLambda^{-1}\Bigl({\bf P}^{-1}{\rm E}[{\bf s}^*{\bf s}^T]\\ -\frac{1}{2}\bigl({\bf M}^T{\bf P}^{-*}+{\bf P}^{-1}{\bf M}\bigr){\rm E}[{\bf s}{\bf s}^T]\Bigr)\bLambda=\\
    \bLambda^{-1}\Bigl({\bf P}^{-1}\bSigma^* -\frac{1}{2}\bigl({\bf M}^T{\bf P}^{-*}+{\bf P}^{-1}{\bf M}\bigr)\bGamma\Bigr)\bLambda={\bf I}_K,
\end{multline}
where we have used \eqref{AC:identity}. Now, \eqref{eq:nugauss} follows since $\mu_k={\bf e}_k^H \boldsymbol{\Xi}{\bf e}_k$. 

It is easily seen that, when $\bSigma$, $\bGamma$, and $\bLambda$ are replaced, respectively, by their estimates $\widetilde\bSigma$, $\widetilde\bGamma$, and $\widetilde\bLambda$ (assuming that $\widetilde\bSigma^{-1}$ and $\widetilde{\bf P}^{-1}$ exist), we can follow the same steps to prove \eqref{eq:nugaussest}. Provided that $\hat{\rm E}[{\bf s}^*{\bf s}^T]$ is in \eqref{AC:auxmatrix} replaced by $\widetilde\bSigma^*$ and $\hat{\rm E}[{\bf s}{\bf s}^T]$ is replaced by 
$\widetilde\bGamma$, the assertion that $\hat\nu_k=1$ follows.

Finally, \eqref{eq:rhogauss} and \eqref{eq:rhogaussest} readily follow by considering the Wirtinger derivative of \eqref{AC:modelscore} by ${\bf s}^*$. \hfill\rule{1.2ex}{1.2ex}


\end{document}